\documentclass[a4paper,11pt]{article}
\usepackage{amsmath}
\usepackage{amsfonts}
\usepackage{amssymb}
\usepackage{wasysym}
\usepackage{amsthm}
\usepackage{tikz-cd}
\usepackage{longtable}
\usepackage{empheq}
\usepackage{ytableau}
\usepackage{url}

\usepackage[hidelinks,colorlinks=true,unicode]{hyperref}
\hypersetup{
	linkcolor=blue,
	citecolor=blue,
	filecolor=magenta,
	urlcolor=blue,
}
\usepackage[left=2cm,right=2cm,top=2cm,bottom=2cm,bindingoffset=0cm]{geometry}

\newtheorem{statement}{Proposition}
\newtheorem{defin}{Definition}
\newtheorem{conj}{Conjecture}
\newtheorem{cor}{Corollary}
\newtheorem{lemma}{Lemma}
\newtheorem{rem}{Remark}

\title{{\bf  Interplay between symmetries of quantum 6j-symbols and the eigenvalue hypothesis}\date{} \vspace{.5cm}}
\author{{\bf Victor Alekseev$^{a,b,c}$\thanks{alekseev.va@phystech.edu},
		Andrey Morozov$^{a,b,c}$\thanks{Andrey.Morozov@itep.ru},
		Alexey Sleptsov$^{a,b,c}$\thanks{sleptsov@itep.ru}}}

\begin{document}
	\maketitle
	
	\vspace{-7cm}
	\begin{center}
	\hfill	ITEP-TH-24/19\\
	\hfill	IITP-TH-16/19\\
	\hfill	MIPT-TH-14/19
	\end{center}
	\vspace{5.5cm}
	
	\vspace{-1cm}
	\begin{center}
		$^a$ {\small {\it ITEP, Moscow 117218, Russia}}\\
		$^b$ {\small {\it Institute for Information Transmission Problems, Moscow 127994, Russia}}\\
		$^c$ {\small {\it Moscow Institute of Physics and Technology, Dolgoprudny 141701, Russia}}\\
	\end{center}
	\vspace{1cm}
	
\begin{abstract}
The eigenvalue hypothesis claims that any quantum Racah matrix for finite-dimensional representations of $U_q(sl_N)$ is uniquely determined by eigenvalues of the corresponding quantum $\cal{R}$-matrices. If this hypothesis turns out to be true, then it will significantly simplify the computation of Racah matrices. Also, due to this hypothesis various interesting properties of colored HOMFLY-PT polynomials will be proved. In addition, it allows one to discover new symmetries of the quantum 6j-symbols, about which almost nothing is known for $N>2$, with the exception of the tetrahedral symmetries, complex conjugation and transformation $q \longleftrightarrow q^{-1}$.

In this paper we prove the eigenvalue hypothesis in $U_q(sl_2)$ case and show that it is equivalent to 6j-symbol symmetries (the Regge symmetry and two argument permutations). Then we apply the eigenvalue hypothesis to inclusive Racah matrices with 3 symmetric incoming representations of $U_q(sl_N)$ and an arbitrary outcoming one. It gives us 8 new additional symmetries that are not tetrahedral ones. Finally, we apply the eigenvalue hypothesis to exclusive Racah matrices with symmetric representations and obtain 4 tetrahedral symmetries.
\end{abstract}
\section{Introduction}

Nowadays we have a number of theories in mathematical and theoretical physics that use Racah coefficients. Sometimes they are referred as Wigner 6j-symbols, that differ by the normalization factor (see equation (\ref{6j_def})). The basic example of Racah coefficients' occurrence is in combining three angular momenta in quantum mechanics. Here they appear as elements of the transformation matrix between two canonical bases corresponding to the different combining orders. Also, Racah coefficients are a powerful instrument that can be used to calculate observables in Chern-Simons theory. Besides, it may also be viewed as elements of the duality matrix between two equivalent $U_q(sl_2)_k$ Wess-Zumino Witten conformal blocks where the quantum deformation $q$ is taken as $k {+} 2$-th root of unity.

When it comes to the knot theory, one of the fundamental ideas there is the concept of knot invariants. It turns out that Racah coefficients play an important role in it, particularly in the wide range of knot polynomials such as HOMFLY-PT \cite{HOMFLY_O}\cite{HOMFLY_OO}\cite{HOMFLY}. One of the most powerful techniques to obtain knot invariants is the Reshetikhin-Turaev approach \cite{RT}, based on quantum groups theory and quantum $\hat{\mathcal{R}}$-matrices and the essential element of this method is a Racah coefficient.

It is well known \cite{KR} that in $U_q(sl_2)$ there is an analytic expression for arbitrary Racah coefficient in terms of a $q$-hypergeometric function ${}_4\Phi_3$. It allows us to investigate 6j-symbols analytically, what leads to different interesting results. There are some rather recent papers \cite{mass1}\cite{mass2}\cite{mass3}\cite{mass4} as an illustration. For algebras with higher ranks $N>2$ the situation is much more complicated. For symmetric representations explicit answers in terms ${}_4\Phi_3$ of were obtained in \cite{3SB, alekseev2020multiplicity}, but for arbitrary $U_q(sl_N)$ representations similar analytic expressions for 6j-symbols are yet to be discovered.   Despite the fact that we can obtain a Racah coefficient value via highest weights method \cite{hw_method}, it is a very complicated and cumbersome approach which quickly becomes practically non-applicable as we try to consider higher representations. Nevertheless, it is very important to generalize various properties of $U_q(sl_2)$ Racah coefficients to the $U_q(sl_N)$ case, which could lead us to the general $U_q(sl_N)$ answers for Racah coefficients in the future.

All symmetries of $U_q(sl_2)$ Racah matrices are well known and well studied. In the present paper we are interested in linear symmetries of the $U_q(sl_N)$ Racah coefficients implied by the eigenvalue hypothesis. Non-linear symmetries (e.g. the pentagon relation), that are more complicated, are out of the scope of this paper. Linear symmetries of $U_q(sl_2)$ Racah coefficients include Regge symmetries, the tetrahedral symmetries and transformation $q \longleftrightarrow q^{-1}$.  However, there should be non-trivial generalizations of Regge and tetrahedral symmetries for $U_q(sl_N)$ that are still not known.

The eigenvalue hypothesis originates from the Yang-Baxter equation for knots or links. This equation in terms of $\hat{\mathcal{R}}$-matrices is nothing but the algebraic form of the third Reidemeister move in knot theory. Each $\hat{\mathcal{R}}$-matrix acts in the tensor product of representations' domain and permutes two adjacent ones, this operator corresponds to the generator of the braid group in knot theory. By diagonalizing $\hat{\mathcal{R}}$-matrices via Racah matrices, one can get the equation defining Racah matrices through $\hat{\mathcal{R}}$-matrices' eigenvalues. This leads to the eigenvalue conjecture, which states that Racah matrices are fully determined by the sets of corresponding $\hat{\mathcal{R}}$-matrices.

This by no means is a trivial fact. In fact the Yang-Baxter equation if solved with respect to Racah matrices has several solutions and the number of solutions becomes larger as the sizes of matrices grow. Nevertheless, it seems that the Racah matrices themselves are always uniquely defined by the eigenvalues of the $\hat{\mathcal{R}}$-matrices or at least it is so in all studied examples. Even more, there is an exact expression for the Racah matrices through the $\hat{\mathcal{R}}$-matrix eigenvalues for the matrices of the size up to $5\times 5$ \cite{Ev_Hyp} and $6\times 6$ \cite{Universality}. The eigenvalue conjecture can be continued even further to include links. Since there are several different diagonal $\hat{\mathcal{R}}$-matrices, Racah matrices depend on eigenvalues of all of them. The exact expressions have been constructed for the link Racah matrices of the size $3\times 3$ as well \cite{cabling}. Another generalization which can be done is to move from 3-strand matrices to a higher number of strands. There, as it appears, one has to use not only the Yang-Baxter equation but the commutation relations on different $\hat{\mathcal{R}}$-matrices, different braid group generators. This also allows to construct the exact expression for corresponding Racah matrices through the eigenvalues of the $\hat{\mathcal{R}}$-matrices at least for the $5\times 5$ matrices in the 4-strand braid \cite{Multistrand}.

Although the eigenvalue conjecture is not proven yet, it was checked in many cases and used, for example, to calculate HOMFLY-PT polynomials for 3-strand knots \cite{Ev_Hyp} and links \cite{3SB} in arbitrary symmetric representations. These cases are special since there are no multiplicities - no coinciding eigenvalues in $\hat{\mathcal{R}}$-matrices. If there are coinciding eigenvalues situation becomes more difficult because there is an additional freedom in the Yang-Baxter equations and its solutions. However even in this case it seems that often Racah matrix can be made block-diagonal with blocks themselves satisfying the eigenvalue conjecture \cite{Bishler}. Another important application of the eigenvalue conjecture is its connection \cite{Alexander} to the known property of Alexander polynomials, which relates all Alexander polynomials for the same knot and different representations, denoted by the hook Young diagrams. This property leads to very unexpected relation of colored Alexander polynomials with the KP integrable hierarchy  \cite{MISHNYAKOV2021115334}. Also the eigenvalue conjecture predicts new symmetries of colored HOMFLY-PT polynomials, which can be checked for some particular examples \cite{mishnyakov2021new, mishnyakov2020novel}.

It is worth mentioning that Racah matrices may be considered in two ways: those whose first three representations tensor product decomposes in the fourth one are called \textit{inclusive} or mixing matrices. On the other hand, there is a definition of the Racah matrix, in which representations are divided into two pairs, and the tensor product of the first pair is transformed by the Racah matrix into the second product, we call such Racah matrices \textit{exclusive}. Obviously, we can rewrite an exclusive Racah matrix using notations for inclusive ones, but exclusive ones will have one representation conjugated, so for different $N$ in $ U_q(sl_N) $ the representations are different.  The most important difference between these two types is that for a sufficiently large $ N $ for any algebra $U_q(sl_N) $, the inclusive Racah matrices do not depend on $ N $. The exclusive Racah matrices, on the contrary, always explicitly depend on $ N $.

\bigskip

In this paper we consider multiplicity-free $U_q(sl_N)$ Racah matrices to find new symmetries for both inclusive and exclusive types. Multiplicity-free means that each tensor product of pairs of considered representations does not contain in its decomposition repeated summands. The method is based on \cite{sleptsov_new_sym}, where the eigenvalue hypothesis is used to predict the equality of particular Racah matrices. In section \ref{2} the eigenvalue hypothesis is reformulated for the purposes of our paper. We start from the general form of the hypothesis and then confine to the particular class of Racah matrices. In fact, 3 symmetric incoming representations and arbitrary outcoming one are considered in section \ref{S3}. This allows us to reduce the hypothesis to a system of linear equations. As a result, all predicted symmetries are listed as the solutions of the system.

In section \ref{S4} we give a proof of the eigenvalue hypothesis for the $U_q(sl_2)$ case. It's well known that there are 144 symmetries for $U_q(sl_2)$ 6j-symbols \cite{klimyk}, but in terms of Racah matrices these relations equate some particular matrix elements, not necessary whole matrices. If we consider only matrix symmetries of 6j-symbols, there are 8 equivalent ones. Racah matrices also have these symmetries because the normalization factors are the same for both sides of equations. All these relations are obtained via the eigenvalue hypothesis for the $U_q(sl_2)$ case. It is also true that the eigenvalue hypothesis conditions are satisfied for $U_q(sl_2)$ Racah matrices that are equal due to symmetries. That means in the $U_q(sl_2)$ case the eigenvalue hypothesis is proven.

Then in section \ref{S5} the same procedure is applied for representations of $U_q(sl_{N>2})$, where the same number of relations arises -- 8 symmetries including identity are obtained. There are a few key features that distinguish $N>2$ from $N=2$. First of all, these symmetries are not 6j-symbol symmetries as long as normalization factors may be different after applying a symmetry. Also, the occurrence of a free parameter in these relations is an interesting feature of the discovered symmetries. This parameter can take an arbitrary non-negative integer values, and it allows us to equate an infinite set of Racah matrices. As the derivation was very similar for $N=2$ and $N>2$, we can see the correspondence between symmetries in these two cases and call $U_q(sl_N)$ symmetries by analogy with $U_q(sl_2)$ ones. In particular, Regge symmetry can be easily generalized for that class of Racah matrices. Also, it is known that $U_q(sl_N)$ 6j-symbols have tetrahedral symmetries \cite{tetra}, but obtained symmetries coincide with them only for $N=2$.

Tetrahedral symmetries for $U_q(sl_N)$ relate Racah matrices of the class that differ from the previous section, this class includes exclusive Racah matrices. In fact, in section \ref{S6} we investigate the exclusive class of Racah matrices with two symmetric incoming and outcoming representations. And we find only tetrahedral symmetries. Then we consider a more complicated case in order to demonstrate the flexibility of the eigenvalue hypothesis method. And we obtain 4 new symmetries, which cannot be expressed through tetrahedral ones.

\section{$\hat{\mathcal{R}}$-matrices, Racah coefficients and the eigenvalue hypothesis}\label{2}
The eigenvalue hypothesis \cite{Ev_Hyp} can be obtained from the Yang-Baxter equation for $\hat{\mathcal{R}}$-matrices and written similar to \cite{cabling} in terms of Racah coefficients. In this equation the $\hat{\mathcal{R}}$-matrix is considered to be known whereas the Racah matrix is not. So we consider the Racah matrix as the solution to the Yang-Baxter equation. The problem is that the Yang-Baxter equation has a lot of solutions. By definition, the Racah matrix is a non-degenerate matrix, therefore we have to consider only non-degenerate solutions of the Yang-Baxter equation. Unfortunately, it does not guarantee the uniqueness of the solution. The eigenvalue hypothesis states that the Racah matrix is uniquely determined by the eigenvalues of the corresponding $\hat{\mathcal{R}}$-matrix. In this section we give definitions of $\hat{\mathcal{R}}$-matrices and 6j-symbols, and then we formulate the eigenvalue hypothesis.


\subsection{$\hat{\mathcal{R}}$-matrices and Racah coefficients}

We work with $U_q(sl_N)$ algebra representations denoted by $R_i$, each one acts in the vector space $V_i$. Operators $\hat{\mathcal{R}}_1 \ldots \hat{\mathcal{R}}_{m}$ are called $\hat{\mathcal{R}}$-matrices, they act on a tensor product of $V_i$. By definition, they solve the Yang-Baxter equation, therefore they also can be considered as a representation of the braid group. Every $\hat{\mathcal{R}}_i$ can be written as a combination of $(V_i,V_{i+1})$ permutation  $P$ and a so-called universal $\check{\mathcal{R}}$-matrix \cite{klimyk}:
\begin{align}
\hat{\mathcal{R}}_i=1_{V_1} \otimes 1_{V_2} \otimes\ldots \otimes P\check{\mathcal{R}}_{i,i+1} \otimes \ldots \otimes1_{V_n}
\end{align}

Matrix form of $\hat{\mathcal{R}}_i$ depends on the choice of basis in the order of tensor product of $V_i$'s. The most convenient basis can be constructed using the highest weight vectors. Let us fix the order in the tensor product. By acting with lowering operators on the highest weight vectors we construct a basis in the resulting space, so we will call this basis as $B_{i_1\ldots i_m}$ with indices corresponding to the product
ordering.

Let us choose the basis on $R_1\otimes \ldots \otimes R_m$ which corresponds to the following order in the tensor product:
\begin{equation}
B_{12,3\ldots m} = (\ldots((R_1\otimes R_2)\otimes R_3)\otimes \ldots) \otimes R_m
\end{equation}

$\hat{\mathcal{R}}_1$ in the chosen basis is diagonal, each row and column corresponds to representation $X_\alpha$ in the decomposition
\begin{equation}
R_1\otimes R_2 = \bigoplus_{\alpha} \mathcal{M}_\alpha^{R_1,R_2}\otimes X_\alpha
\end{equation}
If $\dim\mathcal{M}_\alpha^{R_1,R_2} > 1$, then the choice of the basis that diagonalizes $\hat{\mathcal{R}}_1$ is more complex, we have an additional unfixed rotation in the subspace of $X_\alpha$. However, it's always possible to fix it that $\hat{\mathcal{R}}_1$ is diagonal. However, if we consider $\hat{\mathcal{R}}_2$, it can be diagonal only in the basis corresponding to
\begin{equation}
B_{1,23,4\ldots m} = (\ldots\left(R_1\otimes (R_2\otimes R_3)\right)\otimes \ldots) \otimes R_m
\end{equation}
Therefore, in order to diagonalize the matrix $\hat{\mathcal{R}}_2$ we should make a transformation via $U$-matrix, which is the natural isomorphism between the spaces with different tensor product order:
\begin{align}
U:\ \ (R_1 \otimes R_2) \otimes R_3 &\rightarrow R_1 \otimes (R_2 \otimes R_3)
\end{align}
We can rewrite it in irreducible components, where $M$ is the representation multiplicity in the decomposition.	
\begin{equation}
\begin{split}
R_1 \otimes R_2 &= \bigoplus_i M_{X_i}^{R_1,R_2} \otimes X_i\\
R_2 \otimes R_3 &= \bigoplus_j M_{Y_j}^{R_2,R_3} \otimes Y_j
\end{split}
\end{equation}
\begin{equation}
\begin{split}
(R_1 \otimes R_2) \otimes R_3 &=
\bigoplus_{i,k} M_{X_i}^{R_1,R_2} \otimes M_{R_{4_k}}^{X_i,R_3} \otimes R_{4_k}\\
R_1 \otimes (R_2 \otimes R_3) &=
\bigoplus_{j,k} M_{R_{4_k}}^{R_1,Y_j} \otimes M_{Y_j}^{R_2,R_3} \otimes R_{4_k}
\end{split}
\end{equation}
The associativity of vector spaces requires isomorphism $U$ between two fusions. This transformation is defined by the  Racah matrix or Racah-Wigner 6j-symbols.
\begin{defin}
	Racah coefficients are elements of the Racah matrix that is the map:
	\begin{align}\label{U_mat_def}
	U \left[ \begin{matrix}
	R_1 & R_2 \\
	R_3 & {R_4}
	\end{matrix} \right]: \bigoplus_{i} M_{X_i}^{R_1,R_2} \otimes M_{R_{4}}^{X_i,R_3} \rightarrow \bigoplus_{j} M_{R_{4}}^{R_1,Y_j} \otimes M_{Y_j}^{R_2,R_3}
	\end{align}
\end{defin}
\begin{defin}
	Wigner 6j-symbol is the element of a normalized $U$-matrix:
	\begin{align}
	\left\lbrace \begin{matrix}\label{6j_def}
	R_1 & R_2 & X_i\\
	R_3 & R_4 & Y_j
	\end{matrix} \right\rbrace =  \frac{1}{\sqrt{\dim(X_i)\dim(Y_j)}} U_{i,j} \left[ \begin{matrix}
	R_1 & R_2 \\
	R_3 & R_4
	\end{matrix} \right]
	\end{align}
\end{defin}
\subsection{Inclusive and exclusive Racah coefficients}
We divide Racah matrices into two different classes: inclusive one and exclusive one. This classification naturally follows from two different ways of HOMFLY invariant calculations. Following Reshetikhin-Turaev approach \cite{RT}, in the process of knot invariant calculations it is needed to evaluate the matrices for all possible $R_4$. Let us fix first 3 arguments $R_1,R_2,R_3$ in Racah matrix. For each $R_4 \subset R_1 \otimes R_2 \otimes R_3$ we can write down non-trivial Racah matrices as $U \left[ \begin{matrix}
R_1 & R_2 \\
R_3 & {R_4}
\end{matrix} \right]$. These Racah matrices are called inclusive.

On the other hand, there is another way to calculate HOMFLY-PT polynomials, that is based on Wess-Zumino Witten conformal field theory \cite{tetra}. In the case of arborescent links it requires only two Racah matrices, which we call exclusive ones: $	U \left[ \begin{matrix}
R_1 & \overline{R_2} \\
R_3 & {R_4}
\end{matrix} \right]$ and $	U \left[ \begin{matrix}
R_1 & R_2 \\
\overline{R_3} & {R_4}
\end{matrix} \right]$. These Racah matrices use conjugated representations of $U_q(sl_N)$. The main difference between inclusive and exclusive Racah matrices is that inclusive ones stop depending on $N$ when it is sufficiently large. On the other hand, exclusive Racah matrices do depend on $N$, although the dependence is always algebraic in terms of $q$ and $A=q^N$.

\subsection{Eigenvalue hypothesis}

We write down the expressions that lead to the hypothesis. The eigenvalue conjecture originates from the Yang-Baxter equation for links that is the algebraic form of the third Reidemeister move in knot theory. For knots, it's defined by the equation
\begin{equation}\label{10}
\hat{\mathcal{R}}_1\hat{\mathcal{R}}_2\hat{\mathcal{R}}_1 = \hat{\mathcal{R}}_2\hat{\mathcal{R}}_1\hat{\mathcal{R}}_2
\end{equation}
$U$-matrices acts in tensor cube of the representation $R$:
\begin{align}
U:\ \ (R \otimes R) \otimes R &\rightarrow R \otimes (R \otimes R) \hspace{10mm} U^\dagger=U^{-1}
\end{align}
Let us choose the basis in which $\hat{\mathcal{R}}_1$ is diagonal, then $\hat{\mathcal{R}}_2$ may be not diagonal, but we can reexpress it as $\hat{\mathcal{R}}_2 = U^\dagger\hat{\mathcal{R}}_1 U$ where $\hat{\mathcal{R}}_1$ is diagonal. Substituting $\hat{\mathcal{R}}_2$ into (\ref{10}), we obtain:
\begin{align}
\hat{\mathcal{R}}_1 U^{\dagger} {\hat{\mathcal{R}}}_1 U\hat{\mathcal{R}}_1=U^{\dagger} {\hat{\mathcal{R}}}_1 U \hat{\mathcal{R}}_1 U^\dagger {\hat{\mathcal{R}}}_1 U \label{YB_knots}
\end{align}
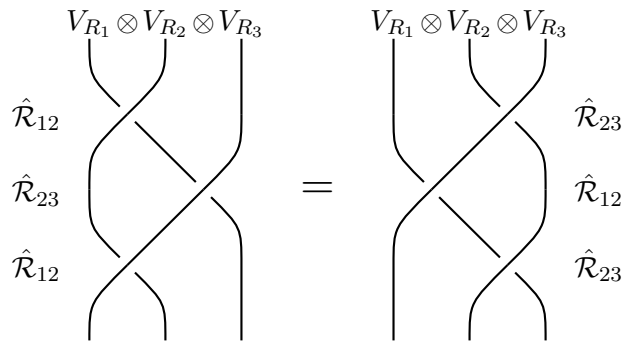
$$\begin{tikzpicture}
\draw[thick](2,0) .. controls (2,1) .. (2,1);
\draw[thick](2,1) .. controls (2,1.5) .. (1.6,1.9);
\draw[thick](1.4,2.1) .. controls (1.4,2.1) .. (0.6,2.9);
\draw[thick](0.4,3.1) .. controls (0,3.5) .. (0,4);
\draw[thick](1,0) .. controls (1,0.5) .. (0.6,0.9);
\draw[thick](0.4,1.1) .. controls (0,1.5) .. (0,2);
\draw[thick](0,2) .. controls (0,2.5) .. (0.5,3);
\draw[thick](0.5,3) .. controls (1,3.5) .. (1,4);
\draw[thick](0,0) .. controls (0,0.5) .. (0.5,1);
\draw[thick](0.5,1) .. controls (0.5,1) .. (1.5,2);
\draw[thick](1.5,2) .. controls (2,2.5) .. (2,3);
\draw[thick](2,3) .. controls (2,3) .. (2,4);
\draw node at (0, 4.2) {$V_{R_1}$};
\draw node at (0.5, 4.2) {$\otimes$};
\draw node at (1, 4.2) {$V_{R_2}$};
\draw node at (1.5, 4.2) {$\otimes$};
\draw node at (2, 4.2) {$V_{R_3}$};
\draw node at (-0.7, 1) {$\hat{\mathcal{R}}_{12}$};
\draw node at (-0.7, 2) {$\hat{\mathcal{R}}_{23}$};
\draw node at (-0.7, 3) {$\hat{\mathcal{R}}_{12}$};
\draw node[scale=1.5] at (3, 2) {=};
\draw node at (4, 4.2) {$V_{R_1}$};
\draw node at (4.5, 4.2) {$\otimes$};
\draw node at (5, 4.2) {$V_{R_2}$};
\draw node at (5.5, 4.2) {$\otimes$};
\draw node at (6, 4.2) {$V_{R_3}$};
\draw node at (6.7, 1) {$\hat{\mathcal{R}}_{23}$};
\draw node at (6.7, 2) {$\hat{\mathcal{R}}_{12}$};
\draw node at (6.7, 3) {$\hat{\mathcal{R}}_{23}$};
\draw[thick](4,0) .. controls (4,1) .. (4,1);
\draw[thick](4,1) .. controls (4,1.5) .. (4.5,2);
\draw[thick](4.5,2) .. controls (4.5,2) .. (5.5,3);
\draw[thick](5.5,3) .. controls (6,3.5) .. (6,4);
\draw[thick](5,0) .. controls (5,0.5) .. (5.5,1);
\draw[thick](5.5,1) .. controls (6,1.5) .. (6,2);
\draw[thick](6,2) .. controls (6,2.5) .. (5.6,2.9);
\draw[thick](5.4,3.1) .. controls (5,3.5) .. (5,4);
\draw[thick](6,0) .. controls (6,0.5) .. (5.6,0.9);
\draw[thick](5.4,1.1) .. controls (5.4,1.1) .. (4.6,1.9);
\draw[thick](4.4,2.1) .. controls (4,2.5) .. (4,3);
\draw[thick](4,3) .. controls (4,3) .. (4,4);
\end{tikzpicture}$$

We can treat this equation as the $U$-matrix defining expression. First of all, we choose the basis in which $\hat{\mathcal{R}}$ is diagonal. The $\hat{\mathcal{R}}$-matrix eigenvalues are well known \cite{Mironov_I}\cite{klimyk} and expressed as the real power of $q$, hence we are able to sort the eigenvalues in descending order of these powers of $q$. Equation (\ref{YB_knots}) is homogeneous with respect to $\hat{\mathcal{R}}$, therefore we can normalize $\hat{\mathcal{R}}$-matrix to make $\det\hat{\mathcal{R}} = \prod_{i}\lambda_i=1$. If these relations are enough to determine the $U$-matrix, then it depends only on the set of normalized eigenvalues. Let us consider two independent Racah matrices $U$ and $\widetilde{U}$, each of them depends on the set of eigenvalues $\hat{\mathcal{R}}$ and $\hat{\widetilde{\mathcal{R}}}$ correspondingly. The eigenvalue hypothesis says that $U$-matrix is fully determined by the set of $\hat{\mathcal{R}}$-matrix normalized eigenvalues.

\begin{conj}[Eigenvalue hypothesis for knot case]\label{C1}
	 Given two equal sets of normalized eigenvalues of two $\hat{\mathcal{R}}$-matrices acting on representation product $R\otimes R\otimes R$, Racah matrices are equal in the corresponding bases, where $\hat{\mathcal{R}}$-matrices are diagonal.
\end{conj}
In this work we are interested in a link case of this relation with a 3-strand braid, where the situation is a bit different. Every strand can carry its own representation in a link whereas in a knot there is only one representation. The initial order of strands may be arbitrary, so 3 different equations arise that correspond to different initial ordering of the representations. There are three equations that can be written down as:
\begin{align}
U_{xyz}:\ \ (R_{x} \otimes R_{y}) \otimes R_{z} &\rightarrow R_{x} \otimes (R_{y} \otimes R_{z}) \hspace{10mm} U_{xyz}^\dagger=U_{xyz}^{-1}\nonumber\\
\hat{\mathcal{R}}_{(xy)z}:\ \ (R_{x}\otimes R_{y})\otimes R_{z} &\rightarrow (R_{y}\otimes R_{x})\otimes R_{z} \nonumber\\
\hat{\mathcal{R}}_{(xy)z} \hat{\mathcal{R}}_{y(xz)} \hat{\mathcal{R}}_{(yz)x} &= \hat{\mathcal{R}}_{x(yz)} \hat{\mathcal{R}}_{(xy)z} \hat{\mathcal{R}}_{z(xy)} \nonumber\\
\hat{\mathcal{R}}_{(yz)x} \hat{\mathcal{R}}_{z(yx)} \hat{\mathcal{R}}_{(zx)y} &= \hat{\mathcal{R}}_{y(zx)} \hat{\mathcal{R}}_{(yx)z} \hat{\mathcal{R}}_{x(yz)}\\
\hat{\mathcal{R}}_{(zx)y} \hat{\mathcal{R}}_{x(zy)} \hat{\mathcal{R}}_{(xy)z} &= \hat{\mathcal{R}}_{z(xy)} \hat{\mathcal{R}}_{(zy)x} \hat{\mathcal{R}}_{y(zx)} \nonumber
\end{align}

Let us choose the basis in which $\hat{\mathcal{R}}_{(xy)z}$ is diagonal, then $\hat{\mathcal{R}}_{x(yz)}$ may be not diagonal, but we can reexpress it as $\hat{\mathcal{R}}_{x(yz)} = U_{xyz}^\dagger\hat{\mathcal{R}}_{x(yz)}'U_{xyz}$ where $\hat{\mathcal{R}}_{x(yz)}'$ is diagonal. Diagonalizing all $\hat{\mathcal{R}}$-matrices, we obtain:
\begin{align}\label{YB_links}
\hat{\mathcal{R}}_{(xy)z}U_{yxz}^{\dagger} {\hat{\mathcal{R}}}_{y(xz)}' \nonumber U_{yzx}\hat{\mathcal{R}}_{(yz)x}&=U_{xyz}^{\dagger} {\hat{\mathcal{R}}}_{x(yz)}' U_{xzy} \hat{\mathcal{R}}_{(xz)y} U_{zxy}^\dagger {\hat{\mathcal{R}}}_{z(xy)}' U_{zyx}\\
\hat{\mathcal{R}}_{(yz)x} U_{zyx}^\dagger{\hat{\mathcal{R}}}_{z(yx)}' U_{zxy}\hat{\mathcal{R}}_{(zx)y}&= U_{yzx}^\dagger {\hat{\mathcal{R}}}_{y(zx)}' U_{yxz} \hat{\mathcal{R}}_{(yx)z} U_{xyz}^\dagger {\hat{\mathcal{R}}}_{x(yz)}' U_{xzy}\\
\hat{\mathcal{R}}_{(zx)y} U_{xzy}^\dagger {\hat{\mathcal{R}}}_{x(zy)}'\nonumber U_{xyz}\hat{\mathcal{R}}_{(xy)z}&= U_{zxy}^\dagger {\hat{\mathcal{R}}}_{z(xy)}' U_{zyx}\hat{\mathcal{R}}_{(zy)x} U_{yzx}^\dagger {\hat{\mathcal{R}}}_{y(zx)}' U_{yxz}
\end{align}

The situation is similar to the knot case, so we can generalize the eigenvalue conjecture for links. It is convenient to use the notation $\mathcal{\hat{R}}_{1} \sim \mathcal{\hat{R}}_{2}$ if two $\mathcal{\hat{R}}$-matrices has the same set of normalized  eigenvalues. Note, that there are 12 different $\mathcal{\hat{R}}$-matrices in (\ref{YB_links}), but only three of them have different sets of eigenvalues. We denote these three $\mathcal{\hat{R}}$-matrices as
$L = (\mathcal{\hat{R}}_1, \mathcal{\hat{R}}_2, \mathcal{\hat{R}}_3)$.
$L$ and $L'$ are called equivalent if $ \forall i \ \mathcal{\hat{R}}_i \sim \mathcal{\hat{R}}_i'$, this can be notated as $L \sim L'$.

 We emphasize the fact that if we have $L\sim L'$, then $L$ consists of a triple of $\mathcal{\hat{R}}$-matrices that differs from the other triple $L'$ only by normalization and a particular unitary rotation.  Hence, it can be said that we use normalization $ \det \mathcal{\hat{R}}_i  = 1$ and look for equal sets of eigenvalues for each $i$. 

\begin{conj}[Eigenvalue hypothesis for link case]\label{C2}
Given two equivalent lists of $\mathcal{\hat{R}}$-matrices $L \sim L'$, the corresponding Racah matrices (that have the same domain) are equal in the bases, where $\hat{\mathcal{R}}$-matrices are diagonal.
\end{conj}
Let us now reformulate the hypothesis in the case of symmetric representations in a more convenient way.

\section{Eigenvalue hypothesis for symmetric representations} \label{S3}

We consider $R_1,R_2,R_3$ to be symmetric representations of $U_q(sl_N)$ and representation $R_4\subset R_1\otimes R_2 \otimes R_3$. We denote by $X,Y,Z$ all irreducible representations that satisfy the following fusion rules:
\begin{align}
\label{tenzexp}
R_1\otimes R_2 =& \mathop{\oplus}\limits_{\alpha=0}^{d{(X)}} X_\alpha \nonumber \\
R_2\otimes R_3 =& \mathop{\oplus}\limits_{\beta=0}^{d{(Y)}} Y_\beta \\
R_1\otimes R_3 =& \mathop{\oplus}\limits_{\gamma=0}^{d{(Z)}} Z_\gamma  \nonumber
\end{align}
In order to get non-zero 6j-symbols we have to impose the following conditions on these representations:
\begin{align}
\label{conds}
R_{4} \subset X_\alpha\otimes R_3\nonumber  \\
R_{4} \subset Y_\beta\otimes R_1 \\
R_{4} \subset Z_\gamma\otimes R_2 \nonumber
\end{align}
Then, depending on $R_4$, not all summands of the RHS of expansions (\ref{tenzexp}) satisfy conditions (\ref{conds}). According to representation theory, the number of irreducible components for each fusion in (\ref{tenzexp}) that satisfies (\ref{conds}) is the same, and we will denote this number as $d$. Note, that in (\ref{tenzexp}) we consider a tensor product of two symmetric representations. Therefore, Young diagrams of the resulting irreducible representations have the number of rows not greater than two. So we can introduce the following notation for Young diagrams of $X,Y,Z$:
\begin{align}\label{xyz}
X_\alpha &= [\mu_{X}-\alpha,\mu_{X} - \delta_X+ \alpha],\ \ 0 \le \alpha \le d \nonumber\\
Y_\beta &= [\mu_{Y}-\beta,\mu_{Y} - \delta_Y+ \beta],\ \ 0 \le \beta \le d\\
Z_\gamma &= [\mu_{Z}-\gamma,\mu_{Z} - \delta_Z+ \gamma],\ \ 0 \le \gamma \le d\nonumber
\end{align}
where $\mu$'s and $\delta$'s are non-negative integers and depend on a particular choice of $R_1,R_2,R_3,R_4$. For $U_q(sl_2)$ representations, we will use variables $x_\alpha = \delta_X-2\alpha$ that corresponds to $X_\alpha$ with columns of the height two removed, the same for $Y$ and $Z$.

As an example of our notation, below we write down the Young diagram sequence with $R_1 = [4], R_2 = [3], R_3 = [2], R_4 = [6,3]$ for $X_\alpha$. The parametrisation is $d=2, \mu_X=6, \delta_X = 5$. The gray color of elements denotes the universal part of all diagrams, while the remaining parts differ among the $\alpha$ range. For $U_q(sl_2)$, $\{x_\alpha\}=\{5,3,1\}$. For $U_q(sl_N)$, sequence $X_\alpha$ is:
\definecolor{gr}{RGB}{180, 180, 180}

$$\ytableausetup{boxsize=1em}\ytableausetup{aligntableaux=center}
\left\{\ytableaushort{}* {6,1}* [*(gr)]{4,1} \ , \ \ytableaushort{}* {5,2}* [*(gr)]{4,1} \ , \ \ytableaushort{}* {4,3}* [*(gr)]{4,1}\right\}$$
\begin{conj}\label{C3}
	If there are two\footnote{All parameters, corresponding to the second matrix we will label by $\ \tilde{ }$.} Racah matrices $U=U\left[ \begin{matrix}
		R_1 & R_2 \\
		R_3 & R_4
	\end{matrix} \right]$ and $\widetilde{U} = U\left[ \begin{matrix}
	\widetilde{R}_1 & \widetilde{R}_2 \\
	\widetilde{R}_3 & \widetilde{R}_4
\end{matrix} \right]$ with symmetric representations $R_i$, $\widetilde{R}_i$, $i=1..3$, specified in (\ref{tenzexp}), (\ref{conds}), (\ref{xyz}) and the following conditions are satisfied:
	\begin{equation}
		\boxed{\label{ev_slN}
			\begin{cases}
				d = \widetilde{d}\\
				\delta_X = \widetilde{\delta}_X\\
				\delta_Y = \widetilde{\delta}_Y\\
				\delta_Z = \widetilde{\delta}_Z\\
		\end{cases}}
	\end{equation}
then corresponding Racah matrices $U$ and $\widetilde U$ are equal.
\end{conj}

\begin{statement} \label{P1}
	Conjecture \ref{C2} in the case of 3 symmetric incoming representations $R_1,R_2,R_3$ is equivalent to Conjecture \ref{C3}.
\end{statement}
\begin{proof}
	Given $R_1,R_2,R_3, R_4$, we are able to find all eigenvalues. Firstly, we find sequences $X_\alpha,Y_\beta,Z_\gamma$, then each representation from these sequences corresponds to an eigenvalue. There are known expression for eigenvalues of $\hat{\mathcal{R}}$-matrices \cite{Mironov_I}\cite{klimyk}:
	\begin{align}
	\begin{split}
	\label{ev_def}
	\lambda_{X_\alpha} &= \epsilon_{X_\alpha}q^{\varkappa(X_\alpha)-\varkappa(R_1) - \varkappa(R_2)} \\
	\varkappa(\lambda) = \sum_{i} \frac{\lambda_i}{2}(\lambda_i&-2i+1) \hspace{10mm} R_1\otimes R_2 = \oplus_{\alpha=0}^d  X_\alpha
	\end{split}
	\end{align}
	where $\epsilon$ is either $1$ or $-1$. It is known from a lot of examples that $\epsilon$ in the case of symmetric representations is just $(-1)^\alpha$. Therefore, for equal number of normalized eigenvalues they always coincide. Thus, below we will neglect the sign since it does not affect the proof.
	
	Let three sets of normalised eigenvalues $ \left( \lambda_{X_\alpha}, \lambda_{Y_\beta}, \lambda_{Z_\gamma} \right) $  be equal to the other sets $ ( \lambda_{\widetilde{X}_\alpha}, \lambda_{\widetilde{Y}_\beta}, \lambda_{\widetilde{Z}_\gamma} ) $. We will consider only one equation $\lambda_{X_\alpha} = \lambda_{\widetilde{X}_\alpha}$, the others can be solved in the same way. If we divide each element of the first sequence by the element from the second sequence, we should get the same value for all $\alpha$. From equation (\ref{ev_def}) we can see that the only variable depending on $\alpha$ is $\varkappa(X_\alpha)$. In other words, normalized eigenvalues are the same if the difference between $\varkappa(X_\alpha)$ and $\varkappa(\widetilde{X}_\alpha)$ does not depend on $\alpha$:
	\begin{equation}
	\varkappa(X_\alpha) - \varkappa(\widetilde{X}_{\alpha}) = \text{const}
	\end{equation}
	If there is more than one eigenvalue in $\hat{\mathcal{R}}$-matrix ($d>0$), this equation gives nontrivial conditions. We can use a monotonic property of the function $\varkappa$ for symmetric representations.  Let us consider an arbitrary symmetric representation product as a sequence of Young diagrams $X_\alpha = [\mu - \alpha,\mu - \delta + \alpha]$ where $0\le\alpha\le d$:
	
	\begin{equation}
	2\varkappa(X_\alpha) = \sum_j \mu_j(\alpha)(\mu_j(\alpha)-2j+1)
	= 2\alpha^2 - 2(1+\delta)\alpha + \text{const}
	\end{equation}
	The function is constantly decreasing from $\alpha=0$ to $\alpha=\frac{1+\delta}{2}$. It is obvious that in Young diagrams $d\le\frac{\delta}{2}$, in the other case the second row will be larger than the first one. In this case all eigenvalues are different. It allows us to write down the ordered sequences of Young diagrams
	$\mu(\alpha)= [\mu- \alpha,\mu-\delta+\alpha] $, $\widetilde{\mu}(\alpha) = [\widetilde{\mu} - \alpha, \widetilde{\mu} - \widetilde{\delta} +\alpha] $. The difference of $\varkappa$ should be constant:
	\begin{equation}
	\varkappa(X_\alpha) - \varkappa(\widetilde{X}_\alpha) = (\widetilde{\delta}- \delta)\alpha + \text{const}
	\end{equation}
	As a result, we get the following situation: if the number of eigenvalues is greater than 1, then the eigenvalue conjecture conditions requires both $\delta$ to be the same for two diagrams $X_0$, $\widetilde{X}_0$. Also, the equality of $d$ is needed to make the number of elements equal. Obviously, given equal $\delta$ and $d$ for the pair of matrices, the hypothesis conditions are satisfied. In $U_q(sl_N),N>2$ case this is equivalent to conditions:
	\begin{equation}
	[\mu,\mu-\delta]=[\widetilde{\mu} +C, \widetilde{\mu}-\widetilde{\delta} +C], \hspace{5mm} d=\widetilde{d}
	\end{equation}
	Or in more compact form:
	\begin{equation}
	\begin{cases}
	d=\widetilde{d}\\
	\delta = \widetilde{\delta}
	\end{cases}
	\end{equation}
	For the other two fusions the proof is the same. $d$ is equal for all three tensor products as it can be shown from representation theory.
\end{proof}
\begin{cor}
	In the $U_q(sl_2)$ case system (\ref{ev_slN}) reduces as follows:
	\begin{equation}
	\begin{cases}
	d = \widetilde{d}\\
	x_0 = \widetilde{x}_{0}\\
	y_0 = \widetilde{y}_{0}\\
	z_0 = \widetilde{z}_{0}
	\end{cases}
	\end{equation}
\end{cor}

Thus, Conjecture \ref{C3} in the case of $U_q(sl_2)$ takes the following form.
\begin{conj}\label{C4}
	If there are two Racah matrices $U=U\left[ \begin{matrix}
	[r_1] & [r_2] \\
	[r_3] & [r_4]
	\end{matrix} \right]$ and $\widetilde{U} = U\left[ \begin{matrix}
	[\widetilde{r}_1] & [\widetilde{r}_2] \\
	[\widetilde{r}_3] & [\widetilde{r}_4]
	\end{matrix} \right]$ with symmetric representations $[r_i]$, $[\widetilde{r}_i]$, $i=1..3$, specified in \eqref{tenzexp}-\eqref{xyz}\footnote{$R_i=[r_i], \ \widetilde{R}_i=[\widetilde{r}_i]$} and the following conditions are satisfied:
	\begin{equation}
	\boxed{\label{ev_sl2}
		\begin{cases}
	d = \widetilde{d}\\
	x_0 = \widetilde{x}_{0}\\
	y_0 = \widetilde{y}_{0}\\
	z_0 = \widetilde{z}_{0}
		\end{cases}}
	\end{equation}
	then corresponding Racah matrices $U$ and $\widetilde U$ are equal.
\end{conj}

This approach also can be applied to the $U_q(sl_N)$ case when representations are symmetric and conjugate to symmetric, namely exclusive case. The analogue of Proposition \ref{P1} for the exclusive case is proven in a similar manner. Instead of parametrization (\ref{xyz}), $X_\alpha$ is defined as $ [\mu_1-\alpha, (\mu_2)^{N-2}, \mu_N+ \alpha ]$, and $\delta$ should be defined as $\mu_1 - \mu_N$, hence we have the corollary as follows.
\begin{cor}
	For $U_q(sl_N)$ exclusive symmetric case with $\delta = \mu_1-\mu_N$, the system is written as follows:
	\begin{equation}
	\begin{cases}
	d = \widetilde{d}\\
	\delta_X = \widetilde{\delta}_X\\
	\delta_Y = \widetilde{\delta}_Y\\
	\delta_Z = \widetilde{\delta}_Z\\
	\end{cases}
	\end{equation}
\end{cor}

\section{Proof of the eigenvalue hypothesis in $U_q(sl_2)$}\label{S4}

For $U_q(sl_2)$ all representations can be labeled by one-row Young diagrams, namely we have $R_1 = [r_1]$, $R_2 = [r_2]$, $R_3 = [r_3]$, $R_4 = [r_4]$ with some integers $r_i$. It is convenient to denote $U_q(sl_2)$ Racah matrix by $
			U\left[ \begin{matrix}
				r_1 & r_2\\
				r_3 & r_4
			\end{matrix} \right] $.
	We want to prove Racah matrix symmetries
	\begin{equation}
		U\left[ \begin{matrix}
			r_1 & r_2\\
			r_3 & r_4
		\end{matrix} \right] = U\left[ \begin{matrix}
			\widetilde{r}_1 & \widetilde{r}_2\\
			\widetilde{r}_3& \widetilde{r}_4
		\end{matrix} \right]
	\end{equation}
	that are implied by the eigenvalue conjecture, i.e. by system (\ref{ev_sl2}). 


We remind that $X_\alpha=[x_\alpha]$, $Y_\alpha=[y_\alpha]$, $Z_\alpha=[z_\alpha]$, $\alpha=0...d,$ are defined in (\ref{tenzexp}), (\ref{conds}) as 3 sequences of representations with the following properties. For each $\alpha$ representations $X_\alpha$, $Y_\alpha$, $Z_\alpha$, on the one hand, arise from tensor product decompositions for $[r_1] \otimes [r_2]$, $[r_2] \otimes [r_3]$ and $[r_1] \otimes [r_3]$ respectively. On the other hand, each decomposition of $[x_\alpha] \otimes [r_3]$, $[y_\alpha] \otimes [r_1]$, $[z_\alpha] \otimes [r_2]$ should contain representation $[r_4]$. Each sequence can be represented as follows:
\begin{equation}
	\begin{split}
		[x_\alpha] &= [x_0-2\alpha],\ \ 0 \le \alpha \le d,\\
		[y_\alpha] &= [y_0 -2\alpha],\ \ 0 \le \alpha \le d,\\
		[z_\alpha] &= [z_0-2\alpha],\ \ 0 \le \alpha \le d
	\end{split}
\end{equation}
 where $x_0,y_0,z_0$ and $d$ depend on $r_1,r_2,r_3,r_4$, $\alpha \in \mathbb{Z}$.

\subsection{Sketch of the proof}
The eigenvalue conjecture is originally formulated in terms of $\mathcal{R}$-matrix eigenvalues, which are given by the quadratic Casimir operator. It reduces to a system of {\it quadratic} equations on parameters of Young diagrams. In the previous section we proved that the eigenvalue conjecture for $U_q(sl_2)$ is equivalently formulated in terms of system \eqref{ev_sl2}, which is {\it linear} on Young diagram parameters.

In this section we prove Conjecture \ref{C4}. Namely, we express parameters $d,x_0,y_0,z_0$ from system \eqref{ev_sl2} in terms of $r_i$. It gives us a system of linear equations on $\widetilde{r}_i$. We solve this system and show that $\widetilde{r}_i$ are related to $r_i$ by the well-known tetrahedral and Regge symmetries of 6j-symbols. In order to simplify the proof let us divide it into 4 steps and sketch them first.

\paragraph{Sketch of the proof.} 

We prove the conjecture in the following steps. 
\begin{enumerate}
	\item Obtain expressions for $d(r_i),x_0(r_i),y_0(r_i),z_0(r_i)$ from the representation theory. Same for  $\widetilde{d}(\widetilde{r}_i)$, etc.

	\item Obtained expressions are given in terms of minima and maxima of sets of linear functions and therefore the dependence on the parameters $r_i$ is piecewise linear. We divide expressions into several cases (\textit{fusion types}), where the dependence is linear.
	
	\item Solve system (\ref{ev_sl2}) for each fusion type separately with respect to $\widetilde{r}_i$.
	\item Identify the solutions $\widetilde{r}_i (r_1,\dots r_4)$ with Racah matrix symmetries. In fact, there are 8 distinct solutions that correspond to 8 symmetries of Racah matrix. \\
\end{enumerate}
\paragraph{Step 1.}
Let us express Racah matrix parameters $d,x_0,y_0,z_0$ in terms of $r_i$. The expressions are obtained from fusion rules coming from the representation theory. On the one hand, $[x_\alpha]$ is obtained from $[r_1]\otimes [r_2]$, hence $\max(r_1-r_2,r_2-r_1) \le x_\alpha \le r_1+r_2$. On the other hand, being multiplied by $[r_3]$ it must give $[r_4]$, hence $\max(r_3-r_4, r_4-r_3) \le x_\alpha \le r_3+r_4$. So, we have
\begin{equation}
	\label{k2}
	\max\left(  \begin{matrix}
		r_1-r_2\\
		r_2-r_1\\
		r_3-r_4\\
		r_4-r_3
	\end{matrix}\right)  \le x_\alpha \le \min\left( \begin{matrix}
		r_1 + r_2 \\
		r_3 + r_4
	\end{matrix}\right), \quad 0\le\alpha \le d.
\end{equation}
Obviously, $x_0=\min\left( 
r_1 + r_2,
r_3 + r_4\right)$ and $d = \min\left(
r_1 + r_2 ,
r_3 + r_4
\right) - 
\max\left( 
|r_1-r_2|,
|r_3-r_4|
\right)$ are piecewise linear functions in $r_i$.



\paragraph{Step 2. }
Now let us divide obtained expressions for $x_0$ and $d$ into several cases, where the dependence in $r_i$ is linear. There are 8 possible cases of the $x$ ranges in total. For a given values  of the representations one of the 8 possible cases takes place. Therefore, we can split all $x$'s into 8 \textit{fusion types} with different minima and maxima expressions, where the type is defined by inequality conditions.


It is clear that 

\begin{equation}\label{d_sl2}
	d = \min\limits_{i}(d_i), \quad \text{where} \quad 2d_i = \left[\begin{matrix}
	r_1+r_2+r_3-r_4,\\
	r_1+r_2-r_3+r_4,\\
	r_1-r_2+r_3+r_4,\\
	-r_1+r_2+r_3+r_4,\\
	2r_4,\\
	2r_3,\\
	2r_2,\\
	2r_1.
	\end{matrix}\right.
\end{equation}
Note that the expression of $x_0$ is fully determined by the expression of $d$ for general choice of $r_i$, that is, $x_0=r_1+r_2$ if $d$ is of types $\{1,2,7,8\}$ and $x_0=r_3+r_4$ otherwise. 

The cases for $y$ and $z$ are similar to the case for $x$:
\begin{equation}\label{k2yz}
		\max\left(  \begin{matrix}
			r_3-r_2\\
			r_2-r_3\\
			r_1-r_4\\
			r_4-r_1
		\end{matrix}\right)  \le y_\alpha \le \min\left( \begin{matrix}
			r_3 + r_2 \\
			r_1 + r_4
		\end{matrix}\right), \qquad \max\left(  \begin{matrix}
			r_1-r_3\\
			r_3-r_1\\
			r_2-r_4\\
			r_4-r_2
		\end{matrix}\right)  \le z_\alpha \le \min\left( \begin{matrix}
			r_1 + r_3 \\
			r_2 + r_4
		\end{matrix}\right)
\end{equation}

As we can see in (\ref{d_sl2}), $d = \min(d_i)$ is invariant under permutations of $r_1,r_2,r_3$. The lengths $d(Y)$, $d(Z)$ of sequences $y$ and $z$ are obtained from two expressions (\ref{k2yz}) that differs from $x$'s case $(\ref{k2})$ only by permutations $r_3\leftrightarrow r_1$ and $r_3\leftrightarrow r_2$ correspondingly. Thus, all three lengths have the same expression $d = d(Y) = d(Z)=\min(d_i)$, as it should be for every Racah matrix according to the representation theory. 

If we consider $y$ and $z$ fusions, there are also 8 possibilities, but they do not add any new fusion types. As we have noted above, $d$ has the same expression $(\ref{d_sl2})$ for $x,y$ and $z$. Additionally, we noticed that the fusion type is enough to determine how inequalities conditions are specified. Hence, the 8 fusion types of $x$ are in one-to-one correspondence with $y$ and $z$ fusion types, and it is consistent to say that the Racah matrix can have one of 8 fusion types. Indeed, once $d$ is fixed for $x$, this also fixes $y$ and $z$ inequalities, so the number of fusion types for Racah matrices is 8. The only difference between $x,y,z$ sequences is in $y_0$ and $z_0$ values. As a result, the complete table of fusions types is expressed in Table \ref{tab_sl2}.
\begin{table}[h]
	\begin{center}
		\begin{tabular}{|c|c|c|c|c|c|}
			\hline
			type & $2d$ & $x_0$ & $y_0$ & $z_0$ & conditions \\
			\hline
			$1$ & $ r_1+r_2+r_3-r_4 $ & $r_1+r_2 $ & $r_2+r_3 $ & $r_1+r_3 $ & $d_1 = \min(d_i)$ \\
			\hline
			$2$ & $ r_1+r_2-r_3+r_4 $ & $r_1+r_2 $ & $r_1+r_4 $ & $r_2+r_4 $ & $d_2 = \min(d_i)$ \\
			\hline
			$3$ & $r_1-r_2+r_3+r_4$ & $r_3+r_4 $ & $r_1+r_4 $ & $r_1+r_3 $ & $d_3 = \min(d_i)$ \\
			\hline
			$4$ & $ -r_1+r_2+r_3+r_4$ & $r_3+r_4 $ & $r_2+r_3 $ & $r_2+r_4$ & $d_4 = \min(d_i)$ \\
			\hline
			$5$ & $2r_4$ & $r_3+r_4 $ & $r_1+r_4 $ & $r_2+r_4 $ & $d_5 = \min(d_i)$ \\
			\hline
			$6$ & $2r_3$ & $r_3+r_4 $ & $r_2+r_3$ & $r_1+r_3 $ & $d_6 = \min(d_i)$ \\
			\hline
			$7$ & $2r_2$ & $r_1+r_2 $ & $r_2+r_3$ & $r_2+r_4 $ & $d_7 = \min(d_i)$ \\
			\hline
			$8$ & $2r_1$ & $r_1+r_2 $ & $r_1+r_4 $ & $r_1+r_3 $ & $d_8 = \min(d_i)$ \\
			\hline
		\end{tabular}
	\end{center}
	\caption{$U_q(sl_2)$ Racah matrix types depending on the value of $d$}\label{tab_sl2}
\end{table}

\paragraph{Steps 3 and 4. }
Now we are ready to solve system \eqref{ev_sl2}. In order to have it in front of our eyes, we write out this system one more time:
\begin{equation}
\label{sys_sl2}
\begin{cases}
d = \widetilde{d},\\
x_0 = \widetilde{x}_0,\\
y_0 = \widetilde{y}_0,\\
z_0 = \widetilde{z}_0.\\
\end{cases}
\end{equation}

First of all, let us consider two particular examples, which perfectly illustrate our idea. For example, the left hand side is of type 1 and the RHS is of type 2:
\begin{equation}
\label{ex1}
\begin{cases}
r_1+r_2+r_3-r_4 = \widetilde{r}_1+\widetilde{r}_2-\widetilde{r}_3+\widetilde{r}_4,\\
r_1+r_2 = \widetilde{r}_1+\widetilde{r}_2,\\
r_2+r_3 = \widetilde{r}_1+\widetilde{r}_4,\\
r_1+r_3 = \widetilde{r}_2+\widetilde{r}_4\\
\end{cases}
\Leftrightarrow \quad
\begin{cases}
\widetilde{r}_1 = r_2,\\
\widetilde{r}_2 = r_1,\\
\widetilde{r}_3 = r_4,\\
\widetilde{r}_4 = r_3.\\
\end{cases}
\end{equation}
It is clear that we get the permutation. Let us consider the second example, when LHS is of type 1 and RHS is of type 5:
\begin{equation}
\label{ex2}
\begin{cases}
r_1+r_2+r_3-r_4 = 2\,\widetilde{r}_4,\\
r_1+r_2 = \widetilde{r}_3+\widetilde{r}_4,\\
r_2+r_3 = \widetilde{r}_1+\widetilde{r}_4,\\
r_1+r_3 = \widetilde{r}_2+\widetilde{r}_4\\
\end{cases}
\Leftrightarrow \quad
\begin{cases}
\widetilde{r}_1 = \frac{-r_1+r_2+r_3+r_4}{2},\\
\widetilde{r}_2 = \frac{r_1-r_2+r_3+r_4}{2},\\
\widetilde{r}_3 = \frac{r_1+r_2-r_3+r_4}{2},\\
\widetilde{r}_4 = \frac{r_1+r_2+r_3-r_4}{2}.\\
\end{cases}
\end{equation}
It is the famous Regge symmetry in this solution. By the way, the previous solution is nothing but a tetrahedral symmetry. In the same way it is easy to check that all 64 solutions of system \eqref{sys_sl2} correspond to Regge and/or tetrahedral symmetries. However let us discuss these symmetries in details and write out all solutions of the system in a compact form.

\begin{defin}
	The equations \eqref{Regge} and their compositions are called Regge symmetries or Regge transformations $(\rho = \frac{r_1 + r_2 + r_3 + r_4}{2}, \rho' = \frac{r_1 + r_3 + r_{12} + r_{23}}{2}, \rho'' = \frac{r_2 + r_4 + r_{12} + r_{23} }{2})$. The first relation also  can be written as the Racah matrices symmetry \cite{KR}.
	\begin{equation}
	\left\lbrace \begin{matrix}\label{Regge}
	r_1 & r_2 & r_{12} \\
	r_3 & r_4 & r_{23}\end{matrix} \right\rbrace = \left\lbrace \begin{matrix}
	\rho - r_3 & \rho - r_4 & r_{12} \\
	\rho - r_1 & \rho - r_2 & r_{23} \\
	\end{matrix} \right\rbrace = \left\lbrace \begin{matrix}
	\rho' - r_3 & r_2 & \rho' - r_{23} \\
	\rho' - r_1 & r_4 & \rho' - r_{12} \\
	\end{matrix} \right\rbrace = \left\lbrace \begin{matrix}
	r_1 & \rho'' - r_4 & \rho'' - r_{23}\\
	r_3 & \rho'' - r_2 & \rho'' - r_{12} \\
	\end{matrix} \right\rbrace
	\end{equation}
\end{defin}

\begin{defin}
	Tetrahedral symmetry is the known property of 6j-symbols \cite{KR,SBVD}. For  $U_q(sl_2)$, it is expressed as argument permutations:
	\begin{align}
	\left\{ \begin{matrix}
	r_1 & r_2 & r_{12} \\
	r_3 & r_4 & r_{23} \end{matrix} \right\}
	= \left\{ \begin{matrix}
	{r_3} & {r_4} & {r_{12}} \\
	{r_1} & {r_2} & {r_{23}} \end{matrix} \right\}  =\left\{ \begin{matrix}
	{r_2} & {r_1} & {r_{12}} \\
	{r_4} & {r_3} & {r_{23}} \end{matrix} \right\} = \left\{ \begin{matrix}
	{r_3} & {r_2} & {r_{23}} \\
	{r_1} & {r_4} & {r_{12}} \end{matrix} \right\} = \left\{ \begin{matrix}
	{r_1} & {r_{12}} & {r_2} \\
	{r_3} & {r_{23}} & {r_4} \end{matrix} \right\}\nonumber
	\end{align}
\end{defin}
However we do not need all these symmetries, but only those that relate the Racah matrices, that is, the symmetries should not affect the matrix indexes $r_{12}$ and $r_{23}$. Therefore, all symmetries interesting for us in matrix terms look like
	\begin{align}
	U\left[ \begin{matrix}
	r_1 & r_2 \\
	r_3 & r_4
	\end{matrix} \right] = U\left[ \begin{matrix}
	r_4 & r_3 \\
	r_2 & r_1
	\end{matrix} \right] &= U\left[ \begin{matrix}
	r_3 & r_4 \\
	r_1 & r_2
	\end{matrix} \right]  = U\left[ \begin{matrix}
	r_2 & r_1 \\
	r_4 & r_3
	\end{matrix} \right]=\\ =
	U\left[ \begin{matrix}\nonumber
	\rho - r_1 & \rho - r_2 \\
	\rho - r_3 & \rho - r_4 \\
	\end{matrix} \right]= U\left[ \begin{matrix}
	\rho - r_4 & \rho - r_3 \\
	\rho - r_2 & \rho - r_1 \\
	\end{matrix} \right]  &= U\left[ \begin{matrix}
	\rho - r_3 & \rho - r_4 \\
	\rho - r_1 & \rho - r_2 \\
	\end{matrix} \right] = U\left[ \begin{matrix}
	\rho - r_2 & \rho - r_1 \\
	\rho - r_4 & \rho - r_3 \\
	\end{matrix} \right].\end{align}

It is easy to see that these symmetries form a group, which we denote by $\Omega$.
\begin{defin} 	The group $\Omega$ of Racah matrix symmetries is defined as follows:
	\begin{equation}
	\begin{split}
	&\Omega \cong \mathbb{Z}_2 \times \mathbb{Z}_2 \times \mathbb{Z}_2, \qquad \Omega = \left<\operatorname{Id}, \omega_1, \omega_2, \omega_3\right>,\\
	&\begin{cases}
	\omega_1: (r_1,r_2,r_3,r_4) \rightarrow (r_3,r_4,r_1,r_2),\\
	\omega_2: (r_1,r_2,r_3,r_4) \rightarrow (r_2,r_1,r_4,r_3),\\
	\omega_3: (r_1,r_2,r_3,r_4) \rightarrow (\rho -r_1,\rho -r_2,\rho -r_3,\rho -r_4), \quad \rho = \frac{r_1+r_2+r_3+r_4}{2}
	\end{cases}\\
	&\omega_1\omega_2 = \omega_2 \omega_1,  \quad \omega_1^2 =\operatorname{Id},    \\
	&\omega_2\omega_3 = \omega_3 \omega_2, \quad  \omega_2^2 =\operatorname{Id},    \\
	&\omega_3\omega_1 = \omega_1 \omega_3,  \quad \omega_3^2 =\operatorname{Id}
	\end{split}
	\end{equation}
\end{defin}
Now one can see that our first example \eqref{ex1} corresponds to $\omega_2$ and our second example \eqref{ex2} corresponds to $\omega_3$. Finally, let us write out all solutions of system \eqref{sys_sl2} in terms of symmetry group $\Omega$:
\begin{table}[h]
\begin{center}
\begin{tabular}{|c|c|c|c|c|c|c|c|c|}
\hline
Type $\setminus$ Symmetry  & $\operatorname{Id}$ & $\omega_1$ & $\omega_2$ & $\omega_3$ &  $\omega_1\omega_2$ &  $\omega_1\omega_3$ &  $\omega_2\omega_3$ & $\omega_1\omega_2\omega_3$ \\
\hline
1&1 & 3&2&5&4&7&6&8 \\
\hline
2&2&4&1&6&3&8&5&7 \\
\hline
3&3&1&4&7&2&5&8&6 \\
\hline
4&4&2&3&8&1&6&7&5 \\
\hline
5&5&7&6&1&8&3&2&4 \\
\hline
6&6&8&5&2&7&4&1&3 \\
\hline
7&7&5&8&3&6&1&4&2 \\
\hline
8&8&6&7&4&5&2&3&1 \\
\hline
\end{tabular}
\end{center}
\caption{Solutions of system \eqref{sys_sl2}}\label{sol_sl2}
\end{table}

\begin{rem}
It is straightforward to check that any degeneracy like $d_i=d_j=\min\limits_{k}(d_k)$ for some $i$ and $j$ does not lead to ambiguity of the choice of $x_0,y_0,z_0$ from Table \ref{tab_sl2}. 
\end{rem}

\section{Symmetries in the $U_q(sl_N)$ symmetric case}\label{S5}
For $U_q(sl_N)$ let us consider a couple of Racah matrices for symmetric representations $R_1=[r_1]$, $R_2=[r_2]$, $R_3=[r_3]$, $R_4=[m_1,m_2,m_3]$, $\widetilde{R}_1=[\widetilde{r}_1]$, $\widetilde{R}_2=[\widetilde{r}_2]$, $\widetilde{R}_3=[\widetilde{r}_3]$, $\widetilde{R}_4=[\widetilde{m}_1, \widetilde{m}_2, \widetilde{m}_3]$, where $r_1,r_2,r_3, m_1,m_2,m_3$ are integers that denote the length of a row in a Young diagram. Also sometimes we will mention $[m_1,m_2,m_3]$ as $r_4$ and $X,Y$ and $Z$ as they are defined in (\ref{tenzexp},\ref{conds}).
\begin{equation}
U\left[ \begin{matrix}
[r_1] & [r_2]\\
[r_3] & [m_1,m_2,m_3]
\end{matrix} \right] = U\left[ \begin{matrix}
[\widetilde{r}_1] & [\widetilde{r}_3] \\
[\widetilde{r}_3] & [\widetilde{m}_1,\widetilde{m}_2,\widetilde{m}_3]
\end{matrix} \right]
\end{equation}
Our aim is to find such $\widetilde{r}_1,\widetilde{r}_2,\widetilde{r}_3,\widetilde{r}_4$ that equality is true for all possible $r_1,r_2,r_3,r_4$ from the eigenvalue hypothesis.

As we have derived in (\ref{ev_slN}), we need to find expressions for $d$ and $\delta$. Let us find the range of $X,Y,Z$ values via fusion rule.
\begin{lemma}
	\label{lemma1}
	Given 3 arbitrary symmetric $U_q(sl_N)$ representations $R_1,R_2,R_3$, representation $R_4 = [m_1,m_2,m_3]\subset R_1\otimes R_2 \otimes R_3$ and $X,Y,Z$ defined as above, then $d,\delta$ fusion types may be expressed as in the Table \ref {tab_slN}.
	
	\begin{table}[h!]\begin{center}
			\begin{tabular}{|c|c|c|c|c|c|}
				\hline
				type & $d$ & $\delta_X$ & $\delta_Y$ & $\delta_Z$ & conditions  \\
				\hline
				$1$ & $ m_1-m_2$ & $a + r_3 $ & $a + r_1 $ & $a + r_2 $ & $d_1 = \min(d_i)$  \\
				\hline
				$2$ & $ r_1-m_3$ & $A - r_3 $ & $a + r_1 $ & $A - r_2 $ & $d_2 = \min(d_i)$  \\
				\hline
				$3$ & $ r_2-m_3$ & $A - r_3 $ & $A - r_1 $ & $a + r_2 $ & $d_3 = \min(d_i)$  \\
				\hline
				$4$ & $ r_3-m_3$ & $a + r_3 $ & $A - r_1 $ & $A - r_2 $ & $d_4 = \min(d_i)$  \\
				\hline
				$5$ & $ m_2-m_3$ & $A - r_3 $ & $A - r_1 $ & $A - r_2 $ & $d_5 = \min(d_i)$  \\
				\hline
				$6$ & $ m_1-r_1$ & $a + r_3 $ & $A - r_1 $ & $a + r_2 $ & $d_6 = \min(d_i)$  \\
				\hline
				$7$ & $ m_1-r_2$ & $a + r_3 $ & $a + r_1 $ & $A - r_2 $ & $d_7 = \min(d_i)$  \\
				\hline
				$8$ & $ m_1-r_3$ & $A - r_3 $ & $a + r_1 $ & $a + r_2 $ & $d_8 = \min(d_i)$  \\
				\hline
			\end{tabular}
		\end{center}
		\caption{$U_q(sl_N)$ $U$-matrix types depending on the value of $d$. Here $a:=2m_1-r_1-r_2-r_3$, $A:=r_1+r_2+r_3-2m_3$ and $r_1+r_2+r_3=m_1+m_2+m_3$ from fusion rules.}\label{tab_slN}\end{table}
\end{lemma}
\begin{proof}
	Let us look at the Young diagrams in decompositions $X_k \subset R_1\otimes R_2$ and $R_4 \subset X_k \otimes R_3$. Littlewood-Richardson rules for the $sl_N$ tensor product decomposition say that in a product of arbitrary Young diagram $\mu$ and symmetric diagram $\nu$ every irreducible component $\rho \subset \mu\otimes \nu$ may be obtained by adding elements of $\nu$ to $\mu$ in such a way that no two boxes of $\nu$ occur in the same column \cite{harris}. We can easily represent the tensor product of symmetric representations as a sum of representations corresponding to two-row Young diagrams $X_k = [r_1+r_2-k, k]$, where $0 \le k \le \min(r_1, r_2)$. On the one hand, $R_4 = [m_1,m_2,m_3]$, but on the other hand we can obtain it from the product of $X_k=[r_1+r_2-k,k]$ and $R_3$. It can be written in the general form as $[r_1+r_2+r_3-k-l-m_3,k+l,m_3]$, where $l$ and $m_3$ are non-negative integers that parametrize tensor product of $X_k\otimes R_3$ and correspond to the number of boxes added to the second and the third rows of $X_k$, whereas $r_3-l-m_3$ boxes are added to the first row of $X_k$.
	
	From the equality of two $R_4$ expressions we immediately find that $l=m_2-k$, $m_3 = r_1+r_2+r_3-m_1-m_2$. Since $l$ is determined by the value of $k$, so $k$ is the only integer parameter in the tensor product, but we will leave this parameter to find the conditions from fusion rules. The $l$ definition requires that
	$l \ge 0$ and $l+m_3 \le r_3$.
	Also, the rules give us additional inequalities $l + k \le r_1+r_2-k$ and $ m_3 \le k$.
	After reducing the inequalities we obtain:
	\begin{equation}\label{kN}
	\max\begin{pmatrix}
	m_3\\ r_1+r_2-m_1
	\end{pmatrix} \le k \le \min\begin{pmatrix}
	r_1+r_2-m_2\\m_2\\r_1\\r_2
	\end{pmatrix}
	\end{equation}
	$d$ can be obtained as the range length, $\delta$ corresponds to the left side of inequalities: $\delta_X = r_1+r_2-2k_{min}$. It's easy to find the expressions for $Y,Z$ due to the symmetry of the problems. We only need to change the variables: $r_3\leftrightarrow r_1$ and $r_3\leftrightarrow r_2$ for $k_{Y}$ and $k_{Z}$ correspondingly. Similar to $U_q(sl_2)$, all fusion types can be expressed as follows:
	\begin{equation}\label{d_slN}
	d_i = \left[\begin{matrix}
	m_1-m_2\\
	r_1-m_3\\
	r_2-m_3\\
	r_3-m_3\\
	m_2-m_3\\
	m_1-r_1\\
	m_1-r_2\\
	m_1-r_3
	\end{matrix}\right.
	\end{equation}
	where $d = \min(d_i)$. Note that $d_X = d_Y = d_Z$ and $\delta_X,\delta_Y,\delta_Z$ expressions may be obtained for every $d_i$, so it is possible to rewrite the inequalities for $X,Y,Z$ as the table of fusion types.
	
\end{proof}

\bigskip

\begin{statement}\label{theor}
	The system (\ref{ev_slN}) for $N>3$ has 8 different solutions that form a group of 3 independent symmetries. Each symmetry is defined for all integer $C\ge C_0$ where $C_0$ is the least number that keeps the number of elements in rows of Young diagrams non-negative.
	{ \begin{empheq}[box=\fbox]{align}\label{sol_slN}
			U\left[ \begin{matrix}
				[r_1] & [r_2]\\
				[r_3] & [m_1,m_2,m_3]\\
			\end{matrix} \right]
			&= U\left[ \begin{matrix}
				[r_1+C] & [r_2+C]\\
				[r_3+C] & \left[m_1+C, m_2+C,m_3+C\right]\\
			\end{matrix} \right] =\\
			&=U\left[ \begin{matrix}\nonumber
				[r_2+C] & [r_1+C]\\
				[m_1{-}m_2{+}m_3{+}C] & \left[ m_1{+}C,r_1{+}r_2{-}m_2{+}C,m_3{+}C\right]\\
			\end{matrix} \right] =\\
			&=U\left[ \begin{matrix}\nonumber
				[r_3+C] & [m_1-m_2+m_3+C]\\
				[r_1 +C] & \left[ m_1+C,r_1+r_3-m_2+C, m_3+C \right]\\
			\end{matrix} \right] =\\
			&=U\left[ \begin{matrix}\nonumber
				[m_1{-}m_2{+}m_3{+}C] & [r_3+C]\\
				[r_2+C] & \left[ m_1{+}C, r_2{+}r_3{-}m_2{+}C, m_3{+}C \right]\\
			\end{matrix} \right] =\\
			&=U\left[ \begin{matrix}\nonumber
				[C-r_1] & [C-r_2]\\
				[C-r_3] & \left[ C-m_3,C-m_2,C-m_1\right]
			\end{matrix}\right]=\\
			&= U\left[ \begin{matrix}\nonumber
				[C-r_2] & [C-r_1]\\
				[C{-}m_1{+}m_2{-}m_3] & \left[C{-}m_3,C{+}m_2{-}r_1{-}r_2,C{-}m_1\right]\\
			\end{matrix}\right]=\\
			&= U\left[ \begin{matrix}\nonumber
				[C-r_3] & [C-m_1+m_2-m_3]\\
				[C-r_1] & \left[C-m_3, C+m_2-r_1-r_3, C-m_1\right]\\
			\end{matrix}\right]=\\
			&= U\left[ \begin{matrix}\nonumber
				[C{-}m_1{+}m_2{-}m_3] & [C-r_3]\\
				[C-r_2] & \left[C-m_3, C{+}m_2{-}r_2{-}r_3, C{-}m_1 \right]\\
			\end{matrix}\right]
	\end{empheq}}
\end{statement}

\begin{rem}
	Also, one can obtain the $U_q(sl_3)$ case from the solution above by removing columns of a height three like $[m_1,m_2,m_3]\rightarrow [m_1-m_3,m_2-m_3]$.
\end{rem}
As an example, we write several equal Racah matrices for both $ sl (3) $ and $U_q(sl_N)$, $N>3$:
\begin{itemize}
	\item $U_q(sl_3)$, $C_1\ge -2$, $C_2 \ge 11$. Note that the inequalities on $C$ are found from the $U_q(sl_N)$ expressions.
	\begin{align}U\left[ \begin{matrix}
	[6] & [5]\\
	[7] & [9,3]\\
	\end{matrix} \right] &= \\
	= U\left[ \begin{matrix}\nonumber
	[6+C_1] & [5+C_1]\\
	[7+C_1] & \left[9, 3
	\right]\\
	\end{matrix} \right]
	=U\left[ \begin{matrix}
	[5+C_1] & [6+C_1]\\
	[8 +C_1] & \left[ 9, 4 \right]
	\end{matrix} \right]
	&= U\left[ \begin{matrix}
	[7+C_1] & [5+C_1]\\
	[6+C_1] & \left[ 9,6 \right]
	\end{matrix}\right]
	= U\left[ \begin{matrix}
	[8+C_1] & [7+C_1]\\
	[5 +C_1] & \left[ 9,5\right]
	\end{matrix}\right]=\\
	=U\left[ \begin{matrix}\nonumber
	[C_2-6] & [C_2-5]\\
	[C_2-7] & \left[ 9,6\right]\\
	\end{matrix} \right]
	=U\left[ \begin{matrix}
	[C_2-5] & [C_2-6]\\
	[C_2-8] & \left[ 9, 5 \right]
	\end{matrix} \right]
	&= U\left[ \begin{matrix}
	[C_2-7] & [C_2-5]\\
	[C_2-6] & \left[9,3 \right]\\
	\end{matrix}\right]
	= U\left[ \begin{matrix}
	[C_2-8] & [C_2-7]\\
	[C_2-5] & \left[ 9,4\right]
	\end{matrix}\right]\end{align}
	\item $U_q(sl_N)$, $N>3$, $C_1\ge -1$, $C_2 \ge 16$.
	\begin{align}U\left[ \begin{matrix}
	[7] & [8]\\
	[11] & [16,9,1]\\
	\end{matrix} \right] &= \\
	= U\left[ \begin{matrix}\nonumber
	[7+C_1] & [8+C_1]\\
	[11+C_1] & \left[16+C_1, 9+C_1,1+C_1\right]\\
	\end{matrix} \right] &=U\left[ \begin{matrix}
	[C_2-7] & [C_2-8]\\
	[C_2-11] & \left[ C_2-1,C_2-9,C_2-16\right]
	\end{matrix}\right]=\\
	=U\left[ \begin{matrix}\nonumber
	[8+C_1] & [7+C_1]\\
	[8+C_1] & \left[ 16+C_1,6+C_1,1+C_1\right]\\
	\end{matrix} \right] &= U\left[ \begin{matrix}
	[C_2-8] & [C_2-7]\\
	[C_2-24] & \left[C_2-1,-6+C_1,C_2-16\right]\\
	\end{matrix}\right]=\\
	=U\left[ \begin{matrix}\nonumber
	[11+C_1] & [8+C_1]\\
	[7 +C_1] & \left[ 16+C_1,9+C_1, 1+C_1 \right]\\
	\end{matrix} \right] &= U\left[ \begin{matrix}
	[C_2-11] & [C_2-8]\\
	[C_2-7] & \left[C_2-1, C_2-9, C_2-16\right]\\
	\end{matrix}\right]=\\
	=U\left[ \begin{matrix}\nonumber
	[8+C_1] & [11+C_1]\\
	[8+C_1] & \left[ 16+C_1, 10+C_1, 1+C_1 \right]\\
	\end{matrix} \right] &= U\left[ \begin{matrix}
	[C_2-8] & [C_2-11]\\
	[C_2-8] & \left[C_2-1, C_2-10, C_2-16 \right]\\
	\end{matrix}\right]
	\end{align}
\end{itemize}
\subsection{Derivation of the symmetries ($\ref{sol_slN}$)}\label{section_sol}
	
	We have to consider all possible values of $r_1,r_2,r_3,r_4$, so there are $8$ different Racah types with different $d$. From Lemma \ref{lemma1} we know the expression for $d,\delta_X,\delta_Y,\delta_Z$, whereas $\delta$'s expressions are determined by the choice of particular $d$, so we have $8$ types for $d$ and $8$ types for $\widetilde{d}$, $64$ cases in total.
	
	We shall denote each system with the types defined in Table \ref{tab_slN} in the following way. A system $c_{ij}$ has type $i$ on the left side and type $j$ on the right side, $1\le i,j \le 8$. For example, a system with $d,\delta_X,\delta_Y,\delta_Z$ from fusion type 2 and $\widetilde{d}, \widetilde{\delta}_X, \widetilde{\delta}_Y, \widetilde{\delta}_Z$ from type 5 is denoted as $c_{25}$.
	In fact, we can just start to solve all 64 systems of equations but there is a more convenient way to solve this system.
	\begin{statement}
		Solutions $c_{ij}$  satisfy the following properties
		 \begin{enumerate}
		 	\item The solution always exists (neglecting inequality conditions) and has one free parameter $C$;
		 	\item If $\widetilde{r}_i$, $\widetilde{m}_i$ is a solution, then $\widetilde{r}_i + C$, $\widetilde{m}_i + C $ also solves the system;
		 	\item Each system from $\{c_{ii},c_{12},c_{13},c_{15}\}$ is inducing the symmetry that can be applied to every type;
		 	\item Symmetry induced by $c_{ij}$ can be expressed as a composition of symmetries induced by $c_{11},c_{12},c_{13}$ and $c_{15}$.
		 \end{enumerate}
	\end{statement}
	\begin{proof}
		
		\
		
	\begin{enumerate}

	\item One can check that $d,\delta_{X},\delta_{Y}$ and $\delta_{Z}$ are linearly independent for every fusion type. In other words, the system is not degenerate, consequently there is a solution for every $d$. Since we have 4 equations but 5 independent variables $r_1,r_2,r_3,m_1,m_2$ ($m_3$ is fixed by the number of boxes conservation condition), the solution has one free parameter.
	
	\item One may notice that Table \ref{tab_slN} has one specific property. We will call $r_i,m_i$ \textit{atomic variables}. Each expression of $d$ and $\delta$ is a linear combination of atomic variables with integer coefficients. it can be seen that the sum of positive coefficients is equal to the sum of negative ones. If one increase all atomic variables by $C$, the values in Table \ref{tab_slN} does not change. Therefore, each solution has the free parameter that is added to $r_i,m_i$.
	
	Let us write the $c_{11}$ solution as an example, we will call it the \textit{basic solution}. The system for this type is:
	\begin{equation}
	\begin{cases}
	m_1-m_2 = \widetilde{m}_1-\widetilde{m}_2\\
	2m_1-r_1-r_2 = 2\widetilde{m}_1-\widetilde{r}_1-\widetilde{r}_2\\
	2m_1-r_2-r_3 = 2\widetilde{m}_1-\widetilde{r}_2-\widetilde{r}_3\\
	2m_1-r_1-r_3 = 2\widetilde{m}_1-\widetilde{r}_1-\widetilde{r}_3
	\end{cases}
	\end{equation}
	Obviously, the system will be satisfied if $r_i=\widetilde{r}_i, m_i=\widetilde{m}_i, 0\le i\le3$. Also we can notice that each side of equations is a substitution of atomic variables, $C$ occurs as an additive constant to the atomic variables. The basic solution is:
	\begin{align}\label{basic}
	U\left[ \begin{matrix}
	[r_1] & [r_2]\\
	[r_3] & [m_1,m_2,m_3]\\
	\end{matrix} \right] =
	U\left[ \begin{matrix}
	[r_1+C] & [r_2+C]\\
	[r_3+C] & \left[m_1+C, m_2+C,m_3+C\right]\\
	\end{matrix} \right]
	\end{align}
	
	As long as we know that the free parameter occurs in a solution as an additive constant, a system $c_{ii}$ with coinciding sides of equations has the same solution. Therefore, ($\ref{basic}$) satisfies not only system $c_{11}$, but  $\forall i \ c_{ii}$. In other words, every Racah matrix can be transformed as in equation ($\ref{basic}$), hence it may be applied for all possible types and, consequently, for all $r_i,m_i$ without any inequality restrictions.
	
	\item Solving systems $c_{12},c_{13},c_{15}$ and omitting inequality conditions on $d$, we get
	\begin{gather}\label{cases_sol}
	c_{15}:\begin{cases}
	r_1=C-\widetilde{r}_1\\
	r_2=C-\widetilde{r}_2\\
	r_3=C-\widetilde{r}_3\\
	m_1=C-\widetilde{m}_3\\
	m_2=C-\widetilde{m}_2\\
	m_3=C-\widetilde{m}_1\\
	\end{cases}
	c_{12}:\begin{cases}
	r_1=\widetilde{m}_1-\widetilde{m}_2+\widetilde{m}_3+C\\
	r_2=\widetilde{r}_3+C\\
	r_3=\widetilde{r}_2+C\\
	m_1=\widetilde{m}_1+C\\
	m_2=\widetilde{r}_2+\widetilde{r}_3-\widetilde{m}_2+C\\
	m_3=\widetilde{m}_3+C\\
	\end{cases}\hspace{-5mm}
	c_{13}:\begin{cases}
	r_1=\widetilde{r}_3+C\\
	r_2=\widetilde{m}_1-\widetilde{m}_2+\widetilde{m}_3+C\\
	r_3=\widetilde{r}_1+C\\
	m_1=\widetilde{m}_1+C\\
	m_2=\widetilde{r}_1+\widetilde{r}_3-\widetilde{m}_2+C\\
	m_3=\widetilde{m}_3+C\\
	\end{cases}
	\end{gather}
	Let us examine $c_{15}$, the other equations can be solved  similarly. The $c_{15}$ solution may be used as an operator $\hat{c}_{15}$ that transform $r_i,m_i$ into $\widetilde{r}_i, \widetilde{m}_i$ that is just the change of variables. As it can be seen from the definition, this transformation changes the expressions of type 1 into the type 5 expressions, but it is still unclear what is going on with inequality conditions. To check this we can substitute variables in the inequalities $d_1 = \min(d_i)$:
	
	\begin{equation}
	d_1 = \min\left(\begin{matrix}
	m_1-m_2 \\
	r_1-m_3 \\
	r_2-m_3 \\
	r_3-m_3 \\
	m_2-m_3 \\
	m_1-r_1 \\
	m_1-r_2 \\
	m_1-r_3 \\
	\end{matrix}\right)\ \longrightarrow \ \widetilde{d}_5 = \min\left( \begin{matrix}
	\widetilde{m}_2-\widetilde{m}_3 \\
	\widetilde{m}_1-\widetilde{r}_1 \\
	\widetilde{m}_1-\widetilde{r}_2 \\
	\widetilde{m}_1-\widetilde{r}_3 \\
	\widetilde{m}_1-\widetilde{m}_2 \\
	\widetilde{r}_1-\widetilde{m}_3 \\
	\widetilde{r}_2-\widetilde{m}_3 \\
	\widetilde{r}_3-\widetilde{m}_3
	\end{matrix}\right) = \min\left( \begin{matrix}
	\widetilde{d_5} \\
	\widetilde{d_6} \\
	\widetilde{d_7} \\
	\widetilde{d_8} \\
	\widetilde{d_1} \\
	\widetilde{d_2} \\
	\widetilde{d_3} \\
	\widetilde{d_4}
	\end{matrix}\right)
	\end{equation}
	This transformation preserves the inequalities. Consequently, the $\hat{c}_{15}$ domain is entire type 1 and the codomain is entire type 5.
	
	However, there is an important property of this solution that is the essential one. As we have seen above, $\hat{c}_{15}$ acts on the set of $d$ expressions like a permutation. The full statement is that $\hat{c}_{15}$ acts on the Table \ref{tab_slN} rows as a permutation, moving $d_i, \delta_x, \delta_y, \delta_z$ simultaneously. If we apply $\hat{c}_{15}$ to arbitrary type $i$, $d_i\rightarrow \widetilde{d}_{j(i)}$ and the inequalities will be satisfied too. This is the change of notations, not values, so $d_i = \widetilde{d}_{j(i)}$, $\delta_x = \widetilde{\delta}_x$ and so on, hence this is just $\hat{c}_{ij}$. For $c_{15}$ the type permutation is $(1,2,3,4, 5,6,7,8) \rightarrow (5,6,7,8, 1,2,3,4)$. Therefore, $\hat{c}_{15}$ transforms every Racah matrix into another one with the mentioned change of type and there are no restricting inequalities because inequalities on RHS are equivalent to the LHS ones for every type. As we will show later, $\hat{c}_{15}$ is equal to $\hat{c}_{26}, \hat{c}_{37}$, etc.
	
	The third statement is proved. Below we describe these transformations as symmetries and call them as if they were in $U_q(sl_2)$. Although they are different, they are the straightforward analogues of $U_q(sl_2)$ ones. Let us write down the permutations.
	\begin{enumerate}
		\item $\hat{c}_{ii}$, $1\le i\le 8$ do not change the type;
		\item $\hat{c}_{15}$ is Regge symmetry and swaps $(1,2,3,4) \leftrightarrow (5,6,7,8)$;
		\item $\hat{c}_{12}$ permutes $(r_2,r_3)$, so $(1,3,5,7) \leftrightarrow (2,4,6,8)$;
		\item $\hat{c}_{13}$ permutes $(r_1,r_3)$, so $(1,2,5,6) \leftrightarrow (3,4,7,8)$.
	\end{enumerate}
	
	\item We will perform the composition of derived symmetries in order to get  $c_{11}$ from arbitrary $c_{ij}$. Firstly, we will transform $c_{ij}$ into $1 \le i' , j'\le 4$ using that $\hat{c}_{15}$ moves types from the second half to the first one. If $4 < i, j\le 8$, we apply it for both sides of $c_{ij}$. If only one index is greater than 4, the transformation is needed only to that side. As a result, we obtain $c_{i'j'}$, where $1 \le i', j'\le 4$. Now we do the similar operation to transform  $c_{i'j'}$ into $1 \le i'' , j''\le 2$ using that $\hat{c}_{13}$ moves types $(3,4)$ to $(1,2)$. Then we may use $\hat{c}_{12}$ to get $c_{11}$. So, every $c_{ij}$ can be expressed in basic solution and 3 additional symmetries' composition. Obtained expressions are correct for all $r_i,m_i$ without additional conditions.

	If we look at them as a group of 8 symmetries neglecting $C$ addition, there are identity, 3 independent elements, and 4 more elements can be obtained by compositions. Every transformation being squared gives the basic one, so it is an involution for the particular $C$ ($C = 0$ for $c_{11}$ or $\forall C \ge m_1$ for $c_{15}$). It is very similar to the situation in $U_q(sl_2)$, so we can call new symmetries analogously to $U_q(sl_2)$. Symmetry from $c_{15}$ is clearly the Regge transformation analogue, another two act similar to $U_q(sl_2)$ permutations. In total, we have discovered that 64 cases of $c_{ij}$ are just 8 solutions that split into 8 different types. These symmetries form a group of 8 elements (for a fixed $C$).
\end{enumerate}

\end{proof}
The proof of Statement \ref{theor} also allows us to generalize the Regge symmetry, which in the case of $ sl (N) $ symmetric representations can be written as follows ($C>m_1$):
\begin{equation}
U\left[ \begin{matrix}\label{regge}
[r_1] & [r_2]\\
[r_3] & [m_1,m_2,m_3]\\
\end{matrix} \right] =U\left[ \begin{matrix}
[C-r_1] & [C-r_2]\\
[C-r_3] & \left[ C-m_3,C-m_2,C-m_1\right]
\end{matrix}\right]
\end{equation}
Also, the second and the third symmetry may be seen as a tetrahedral symmetry generalization from $U_q(sl_2)$ for the inclusive class of Racah matrices. For example, $r_1\leftrightarrow r_2$ permutation analogue:
\begin{equation}
U\left[ \begin{matrix}
[r_1] & [r_2]\\
[r_3] & [m_1,m_2,m_3]\\
\end{matrix} \right] = U\left[ \begin{matrix}
[r_2+C] & [r_1+C]\\
[m_1{-}m_2{+}m_3{+}C] & \left[ m_1{+}C,r_1{+}r_2{-}m_2{+}C,m_3{+}C\right]\\
\end{matrix} \right]
\end{equation}
\subsection{Another approach to (\ref{sol_slN}) derivation}\label{26}
There is another way to derive these symmetries. In \cite{3SB} a new connection between symmetric $U_q(sl_N)$ and $U_q(sl_2)$ Racah matrices was derived from eigenvalue hypothesis:
\begin{equation}
U_{U_q(sl_N)}\left[ \begin{matrix}
[r_1] & [r_2]\\
[r_3] & [m_1,m_2,m_3]
\end{matrix} \right] = U_{U_q(sl_2)}\left[ \begin{matrix}
r_1-m_3 & r_2-m_3\\
r_3-m_3 & m_1-m_2
\end{matrix} \right]
\end{equation}
This allows us to derive $U_q(sl_N)$ symmetries as a continuation of known $U_q(sl_2)$ symmetries to arbitrary $N$. The answers obtained from both approaches are the same, because both derivations use the eigenvalue hypothesis, just in a different way. Let us derive Regge symmetry and one permutation, all other symmetries may be obtained in the same way.

Now we can apply Regge symmetry and then reexpress the symbol as $U_q(sl_N)$ one:
\begin{equation}
U_{U_q(sl_2)}\left[ \begin{matrix}
r_1-m_3 & r_2-m_3\\
r_3-m_3 & m_1-m_2
\end{matrix} \right] = U_{U_q(sl_2)}\left[ \begin{matrix}
m_1-r_1 & m_1-r_2\\
m_1-r_3 & m_2-m_3
\end{matrix} \right] = U_{U_q(sl_N)}\left[ \begin{matrix}
[\widetilde{r}_1] & [\widetilde{r}_2]\\
[\widetilde{r}_3] & [\widetilde{m}_1,\widetilde{m}_2,\widetilde{m}_3]
\end{matrix} \right]
\end{equation}
Solving the system of equations for $\widetilde{r}_1,\widetilde{r}_2,\widetilde{r}_3,\widetilde{m}_1,\widetilde{m}_2,\widetilde{m}_3$ with fusion rule conditions we get the following symmetry of $U_q(sl_N)$ Racah matrices:
\begin{equation}
U_{U_q(sl_N)}\left[ \begin{matrix}
[r_1] & [r_2]\\
[r_3] & [m_1,m_2,m_3]
\end{matrix} \right] = U_{U_q(sl_N)}\left[ \begin{matrix}
[C-r_1] & [C-r_2]\\
[C-r_3] & [C-m_3,C-m_2,C-m_1]
\end{matrix} \right]
\end{equation}
This is the same symmetry we derived above. Now let us consider row permutation:
\begin{equation}
U_{U_q(sl_2)}\left[ \begin{matrix}
r_1-m_3 & r_2-m_3\\
r_3-m_3 & m_1-m_2
\end{matrix} \right] = U_{U_q(sl_2)}\left[ \begin{matrix}
r_3-m_3 & m_1-m_2\\
r_1-m_3 & r_2-m_3
\end{matrix} \right] = U_{U_q(sl_N)}\left[ \begin{matrix}
[\widetilde{r}_1] & [\widetilde{r}_2]\\
[\widetilde{r}_3] & [\widetilde{m}_1,\widetilde{m}_2,\widetilde{m}_3]
\end{matrix} \right]
\end{equation}
The solution is:
\begin{equation}
U_{U_q(sl_N)}\left[ \begin{matrix}
[r_1] & [r_2]\\
[r_3] & [m_1,m_2,m_3]
\end{matrix} \right] = U_{U_q(sl_N)}\left[ \begin{matrix}
[r_1+C] & [m_1-m_2+m_3+C]\\
[r_3+C] & [m_1+C,r_1+r_3-m_2+C,m_3+C]
\end{matrix} \right]
\end{equation}
This also coincides with symmetries (\ref{sol_slN}).
\section{Symmetries in the exclusive $U_q(sl_N)$ cases}\label{S6}
Tetrahedral symmetries are widely known for both $U_q(sl_N)$ and $U_q(sl_2)$ cases. However, the known $U_q(sl_N)$ generalization connects only Racah matrices of the particular type, including exclusive ones. This generalization is expressed as follows:
\begin{defin}
	Tetrahedral symmetry is the known property of 6j-symbol to be invariant after transformations \cite{tetra} ($\rho_i,\mu,\nu$ are arbitrary Young diagrams):
	\begin{align}\label{tetra}
	\left\{ \begin{matrix}
	\rho_1 & \rho_2 & \mu \\
	\rho_3 & \rho_4 &\nu \end{matrix} \right\}
	&= \left\{ \begin{matrix}
	\overline{\rho_3} & \overline{\rho_2} &\overline{\nu} \\
	\overline{\rho_1} & \overline{\rho_4} &\overline{\mu}\end{matrix} \right\}
	= \left\{ \begin{matrix}
	\rho_3 & \overline{\rho_4} & \overline{\mu} \\
	\rho_1 & \overline{\rho_2} & \overline{\nu}\end{matrix} \right\}    =\\
	&= \left\{ \begin{matrix}
	{\rho_1} & \overline{\mu} & \overline{\rho_2} \\
	\overline{\rho}_3 & \overline{\nu} & \overline{\rho_4}\end{matrix} \right\}
	=\left\{ \begin{matrix}
	{\rho_2} & {\rho_1} & {\mu} \\
	\overline{\rho}_4 & \overline{\rho_3} & \overline{\nu}\end{matrix} \right\}\nonumber.
	\end{align}
\end{defin}
In this section we investigate whether there are some possibilities for the eigenvalue hypothesis to obtain Racah symmetries with non-symmetric representations using the example of the exclusive Racah matrices. As a result, only tetrahedral symmetries are obtained.
   Then we do the same for a more general class of 6j-symbols and get some new symmetries. Unfortunately, we cannot check these symmetries on particular examples, because corresponding 6j-symbols are still unknown. However, our aim here is to show that the eigenvalue hypothesis can be applied to a wide range of 6j-symbols.

\subsection{$R_4$ is symmetric}
\begin{defin}
	We shall call two 6j-symbols below type I and type II \cite{MFS}.
	\begin{align}
	\text{I type: }\left\lbrace \begin{matrix}
	[r_1] & \overline{[r_2]} & X\\
	[r_3] & [r_4] & Y
	\end{matrix} \right\rbrace \hspace{5mm}
	\text{II type: }
	\left\lbrace \begin{matrix}
	[r_1] & [r_2] & X\\
	\overline{[r_3]} & [r_4] & Y
	\end{matrix} \right\rbrace
	\end{align}
	where $r_1,r_2,r_3,r_4$ are integers that denote numbers of boxes for $U_q(sl_N)$ symmetric representations $R_1,R_2,R_3,R_4$.  $X,Y$ are Young diagrams that satisfy the fusion rules.
\end{defin}
\begin{statement}
	Every exclusive Racah coefficient with symmetric and conjugate to symmetric representations belongs to one of the two kinds: type I or type II.
\end{statement}
Let us consider a couple of type I 6j-symbols where $r_4=r_1+r_3-r_2$.
\begin{equation}
\left\lbrace \begin{matrix}
[r_1] & \overline{[r_2]} & X\\
[r_3] & [r_4] & Y
\end{matrix} \right\rbrace = \left\lbrace \begin{matrix}
[\widetilde{r}_1] & \overline{[\widetilde{r}_2]} & \widetilde{X}\\
[\widetilde{r}_3] & [\widetilde{r}_4] & \widetilde{Y}
\end{matrix} \right\rbrace
\end{equation}
Fusion rules for $X,Y$ and $Z$ are obtained for the more general case in the next subsection, where $R_4 = [k_1+r_2, r_2^{N-2}, k_2]$. For this case one should assume $k_1=r_4,k_2=r_2$ in Table \ref{d_conj}. The system may be solved manually due to a small number of cases. The solutions are:\\
\begin{equation}U\left[ \begin{matrix}
[r_1] & \overline{[r_2]}\\
[r_3] & [r_4]\\
\end{matrix} \right] = U\left[ \begin{matrix}
[r_4] & \overline{[r_3]}\\
[r_2] & [r_1]\\
\end{matrix} \right]= U\left[ \begin{matrix}
[r_2] & \overline{[r_1]}\\
[r_4] & [r_3]\\
\end{matrix} \right]= U\left[ \begin{matrix}
[r_3] & \overline{[r_4]}\\
[r_1] & [r_2]\\
\end{matrix} \right]\end{equation}
It can be easily checked that symmetries above are just tetrahedral symmetry. The situation is the same for type II 6j-symbols.
\subsection{$R_4$ is a combination of symmetric and conjugate to symmetric}
In this subsection we consider another class of 6j-symbols that differs in $R_4$. The key point of this derivation is to show that the eigenvalue hypothesis can give us nontrivial symmetries in a more complex situations than symmetric representations. For that reason we generalized the previous case by replacing $R_4 = [\alpha, \beta^{N-1}]$ with $R_4 = [\alpha, \beta^{N-2}, \gamma]$, where $\alpha,\beta,\gamma$ are some non-negative integers.

Let us consider a couple of generalized type I 6j-symbols where $r_1,r_2,r_3, k_1, k_2$ and $\widetilde{r}_1,\widetilde{r}_2,\widetilde{r}_3, \widetilde{k}_1, \widetilde{k}_2$ are integers that denote the length of the row in a Young diagram. Note that from fusion rules $r_1+r_3=k_1+k_2$ should be satisfied for non-trivial 6j-symbols.
\begin{equation}U\left[ \begin{matrix}
[r_1] & \overline{[r_2]}\\
[r_3] & [k_1+r_2,r_2^{N-2},k_2]\\
\end{matrix} \right] = U\left[ \begin{matrix}
[\widetilde{r}_1] & \overline{[\widetilde{r}_3]}\\
[\widetilde{r}_3] & [\widetilde{k}_1+ \widetilde{r}_2, \widetilde{r}_2^{N-2}, \widetilde{k}_2]
\end{matrix} \right]
\end{equation}
\begin{lemma}
	The Racah types obtained from fusion rules are described in Table \ref{d_conj}.
	\begin{table}[h]\begin{center}
			\begin{tabular}{|c|c|c|c|c|}
				\hline
				type & $d$ & $\delta_X$ & $\delta_Y$ & $\delta_Z$ \\
				\hline
				$1$ & $ r_1 $ & $ r_1+r_2 $ & $ 2k_1+r_2-r_3 $ & $ r_1 + r_3 $ \\
				\hline
				$2$ & $ r_3 $ & $ 2k_1-r_1+r_2 $ & $ r_2 + r_3 $ & $ r_1 + r_3 $ \\
				\hline
				$3$ & $ k_1 $ & $ 2k_1-r_1+r_2 $ & $ 2k_1+r_2-r_3 $ & $ r_1 + r_3 $ \\
				\hline
				$4$ & $ k_2 $ & $ r_1+r_2 $ & $ r_2 + r_3 $ & $ r_1 + r_3 $ \\
				\hline
			\end{tabular}
			\caption{$U_q(sl_N)$ generalized type I $U$-matrix types}\label{d_conj}
	\end{center}\end{table}
\end{lemma}
\begin{proof}
	The derivation of this lemma is the same as in the previous section, so it's omitted.
\end{proof}

The system can be solved manually due to a small number of cases. The solutions are:
\begin{align}
U\left[ \begin{matrix}
[r_1] & \overline{[r_2]}\\
[r_3] & [k_1+r_2,r_2^{N-2},k_2]\\
\end{matrix} \right] &= U\left[ \begin{matrix}
[r_3] & \overline{[r_2+k_1-k_2]}\\
[r_1] & [k_1+r_2,(r_2+k_1-k_2)^{N-2}, k_1]\\
\end{matrix} \right] =\\ =U\left[ \begin{matrix}
[k_1] & \overline{[k_1-r_1+r_2]}\\
[k_2] & [k_1+r_2,(k_1-r_1+r_2)^{N-2},r_3]\\
\end{matrix} \right]  &=U\left[ \begin{matrix}
[k_2] & \overline{[k_1-r_3+r_2]}\\
[k_1] & [k_1+r_2,(k_1-r_3+r_2)^{N-2},r_1]\\
\end{matrix} \right]\nonumber
\end{align}
We always can use $r_1+r_3 = k_1 + k_2$ to substitute $k_2$:
\begin{align}
	\hspace{-15mm}U\left[ \begin{matrix}
		[r_1] & \overline{[r_2]}\\
		[r_3] & [k_1+r_2,r_2^{N-2},r_1+r_3-k_1]\\
	\end{matrix} \right] &= U\left[ \begin{matrix}
		[r_3] & \overline{[r_2+2k_1-r_1-r_3]}\\
		[r_1] & [k_1+r_2,(r_2+2k_1-r_1-r_3)^{N-2}, k_1]\\
	\end{matrix} \right] =\\\hspace{-15mm} =U\left[ \begin{matrix}
		[k_1] & \overline{[k_1-r_1+r_2]}\\
		[r_1+r_3-k_1] & [k_1+r_2,(k_1-r_1+r_2)^{N-2},r_3]\\
	\end{matrix} \right]  &=U\left[ \begin{matrix}
		[r_1+r_3-k_1] & \overline{[k_1-r_3+r_2]}\\
		[k_1] & [k_1+r_2,(k_1-r_3+r_2)^{N-2},r_1]\\
	\end{matrix} \right]\nonumber
\end{align}
As we can see, there are some new relations for type I 6j-symbols. In a similar way relations for type II can be obtained.

\section{Selected results}
\begin{itemize}
	\item For the $U_q(sl_2)$ general case, the eigenvalue hypothesis has been proven, and it is equivalent to the following symmetries $\left(\rho = \frac{r_1 + r_2 + r_3 + r_4}{2}\right)$:
	\begin{align}
	U\left[ \begin{matrix}
	r_1 & r_2 \\
	r_3 & r_4
	\end{matrix} \right] = U\left[ \begin{matrix}
	r_4 & r_2 \\
	r_3 & r_1
	\end{matrix} \right] &= U\left[ \begin{matrix}
	r_3 & r_4 \\
	r_1 & r_2
	\end{matrix} \right]  = U\left[ \begin{matrix}
	r_2 & r_1 \\
	r_4 & r_3
	\end{matrix} \right]=\\ =
	U\left[ \begin{matrix}\nonumber
	\rho - r_1 & \rho - r_2 \\
	\rho - r_3 & \rho - r_4 \\
	\end{matrix} \right]= U\left[ \begin{matrix}
	\rho - r_4 & \rho - r_3 \\
	\rho - r_2 & \rho - r_1 \\
	\end{matrix} \right]  &= U\left[ \begin{matrix}
	\rho - r_3 & \rho - r_4 \\
	\rho - r_1 & \rho - r_2 \\
	\end{matrix} \right] = U\left[ \begin{matrix}
	\rho - r_2 & \rho - r_1 \\
	\rho - r_4 & \rho - r_2 \\
	\end{matrix} \right]\end{align}
	\item For inclusive $U_q(sl_N)$ with symmetric incoming representations, $N>3$, $\forall C\ge C_0$.
{ \begin{align}U\left[ \begin{matrix}
	[r_1] & [r_2]\\
	[r_3] & [m_1,m_2,m_3]\\
	\end{matrix} \right] &=  U\left[ \begin{matrix}
	[r_1+C] & [r_2+C]\\
	[r_3+C] & \left[m_1+C, m_2+C,m_3+C\right]\\
	\end{matrix} \right] \\
	&=U\left[ \begin{matrix}\nonumber
	[r_2+C] & [r_1+C]\\
	[m_1{-}m_2{+}m_3{+}C] & \left[ m_1{+}C,r_1{+}r_2{-}m_2{+}C,m_3{+}C\right]\\
	\end{matrix} \right] =\\
	&=U\left[ \begin{matrix}\nonumber
	[r_3+C] & [m_1-m_2+m_3+C]\\
	[r_1 +C] & \left[ m_1+C,r_1+r_3-m_2+C, m_3+C \right]\\
	\end{matrix} \right] = \\
	&=U\left[ \begin{matrix}\nonumber
	[m_1{-}m_2{+}m_3{+}C] & [r_3+C]\\
	[r_2+C] & \left[ m_1{+}C, r_2{+}r_3{-}m_2{+}C, m_3{+}C \right]\\
	\end{matrix} \right] =\\
	&=U\left[ \begin{matrix}\nonumber
	[C-r_1] & [C-r_2]\\
	[C-r_3] & \left[ C-m_3,C-m_2,C-m_1\right]
	\end{matrix}\right]=\\
	&= U\left[ \begin{matrix}\nonumber
	[C-r_2] & [C-r_1]\\
	[C{-}m_1{+}m_2{-}m_3] & \left[C{-}m_3,C{+}m_2{-}r_1{-}r_2,C{-}m_1\right]\\
	\end{matrix}\right]=\\
	&=U\left[ \begin{matrix}\nonumber
	[C-r_3] & [C-m_1+m_2-m_3]\\
	[C-r_1] & \left[C-m_3, C+m_2-r_1-r_3, C-m_1\right]\\
	\end{matrix}\right]=\\
	 &=U\left[ \begin{matrix}\nonumber
	[C{-}m_1{+}m_2{-}m_3] & [C-r_3]\\
	[C-r_2] & \left[C-m_3, C{+}m_2{-}r_2{-}r_3, C{-}m_1 \right]\\
	\end{matrix}\right]
	\end{align}}
where $r_1+r_2+r_3=m_1+m_2+m_3$.It also can be rewritten as $U_q(sl_3)$ solution in a trivial way.
	\item For exclusive $U_q(sl_N)$ 6j-symbols of type I, the eigenvalue hypothesis predicts only tetrahedral symmetries:
\begin{equation}U\left[ \begin{matrix}
		[r_1] & \overline{[r_2]}\\
		[r_3] & [r_4]\\
	\end{matrix} \right] = U\left[ \begin{matrix}
		[r_4] & \overline{[r_3]}\\
		[r_2] & [r_1]\\
	\end{matrix} \right]= U\left[ \begin{matrix}
		[r_2] & \overline{[r_1]}\\
		[r_4] & [r_3]\\
	\end{matrix} \right]= U\left[ \begin{matrix}
		[r_3] & \overline{[r_4]}\\
		[r_1] & [r_2]\\
	\end{matrix} \right]\end{equation}
	where restriction $r_1+r_3 = r_2 + r_4$ is assumed.
	\item For combined $R_4$, eigenvalue hypothesis predicts:
\begin{align}
	U\left[ \begin{matrix}
		[r_1] & \overline{[r_2]}\\
		[r_3] & [k_1+r_2,r_2^{N-2},k_2]\\
	\end{matrix} \right] &= U\left[ \begin{matrix}
		[r_3] & \overline{[r_2+k_1-k_2]}\\
		[r_1] & [k_1+r_2,(r_2+k_1-k_2)^{N-2}, k_1]\\
	\end{matrix} \right] =\\ =U\left[ \begin{matrix}
		[k_1] & \overline{[k_1-r_1+r_2]}\\
		[k_2] & [k_1+r_2,(k_1-r_1+r_2)^{N-2},r_3]\\
	\end{matrix} \right]  &=U\left[ \begin{matrix}
		[k_2] & \overline{[k_1-r_3+r_2]}\\
		[k_1] & [k_1+r_2,(k_1-r_3+r_2)^{N-2},r_1]\\
	\end{matrix} \right]\nonumber
\end{align}
	where restriction $r_1+r_3 = k_1 + k_2$ is assumed.
\end{itemize}
\section{Conclusion}
The main goal of this paper was to find some general relations for Racah matrices. Indeed, we have discovered a plenty of new symmetries. However, it is hard to say, whether these symmetries occur only in the considered case or it has more general form. As we can see from the second section, a tetrahedral symmetry for $U_q(sl_2)$ case may be generalized in two separate ways: as an $U_q(sl_N)$  tetrahedral symmetry or as a completely new symmetry.

Besides, there are no generalization for Regge symmetry in $U_q(sl_N)$. On the other hand, we have found the $U_q(sl_N)$ symmetry expression for inclusive case with $U_q(sl_N)$ symmetric incoming representations. This symmetry becomes the Regge one for $N=2$. This fact allows us to suggest that the Regge symmetry may be generalized to an arbitrary $U_q(sl_N)$ 6j-symbol.

Relations, discovered in section \ref{S5}, affirm that a $U_q(sl_N)$ multiplicity-free Racah matrix is not changed if a constant integer is added to every row in Young diagrams from the arguments. In other words, this case depends only on the difference between row lengths. That may be seen explicitly in the subsection \ref{26}, but it is possible that this feature is more general and can be applied for arbitrary $U_q(sl_N)$ 6j-symbols.

We should mention that the results are based on the eigenvalue conjecture that is not proven, but only known to be correct for a lot of examples. Nevertheless, it is proven for $U_q(sl_2)$, and, according to our research, we can claim that it is probably correct at least in such simple cases as symmetric incoming representations.

\section*{Acknowledgements}
Our work was partly supported by the grant of the Foundation for the Advancement of Theoretical Physics ``BASIS" (A.M., A.S. and A.V.), 17-01-00585 (A.M.), 18-31-20046 (A.S.), 19-51-50008-Yaf-a (A.M.), 18-51-05015-Arm-a (A.M, A.S.), 18-51-45010-Ind-a (A.M, A.S.), 19-51-53014-GFEN-a (A.M, A.S.), 20-01-00644 (A.M., A.S. and A.V.), by President of Russian Federation grant MK-2038.2019.1 (A.M.). The work was also partly funded by RFBR and NSFB according to the research project 19-51-18006 (A.M.). On behalf of all authors, the corresponding author states that there is no conflict of interest.

\bibliographystyle{unsrturl}
\bibliography{Ev_hyp_syms}{}

\end{document}